\documentclass[sigplan,screen]{acmart}

\usepackage{amsmath,amssymb,xcolor,graphicx,wrapfig,paralist}
\usepackage{algorithmicx,algorithm}
\usepackage[noend]{algpseudocode}
\usepackage{booktabs}
\usepackage{proof}
\usepackage{tikz}
\usepackage{csquotes}
\usepackage{graphics}
\usepackage{stmaryrd}
\usepackage{subcaption}
\usepackage{url}
\usepackage{bbm}

\usepackage{pgfplots}

\usepackage[l2tabu, orthodox]{nag} 

\usepackage{tabu,multirow}
\usepackage{xparse}
\usepackage{mathtools}
\usepackage{environ}

\usepackage[english]{babel}
 \usepackage[latin1]{inputenc}

\usepackage{caption}
\usepackage{subcaption}

\usepackage{algorithm}
\usepackage{algpseudocode}

\usepackage{marginfix}
\usepackage{xifthen}

\usepackage{listings}
\lstdefinestyle{customJ}{
  belowcaptionskip=1\baselineskip,
  breaklines=true,
  frame=L,
  xleftmargin=\parindent,
  language=Java,
  showstringspaces=false,
  basicstyle=\footnotesize\ttfamily,
  keywordstyle=\bfseries\color{green!40!black},
  commentstyle=\itshape\color{purple!40!black},
  identifierstyle=\color{blue},
  stringstyle=\color{orange},
}

\usepackage{afterpage}

\usepackage{xfrac}

\usepackage{balance}
\usetikzlibrary{shapes,positioning,arrows,calc,automata,matrix,fit,arrows.meta}
\tikzstyle{state}+=[minimum size = 6mm, inner sep=0,outer sep=1]
\tikzset{->,>=stealth'}
\definecolor{limeade}{RGB}{102, 153, 0}
\definecolor{crimson}{RGB}{227,23,10}
\definecolor{seafoamgreen}{RGB}{169,229,187}
\definecolor{bananamania}{RGB}{252,246,177}
\definecolor{saffron}{RGB}{247,179,43}
\definecolor{indianred}{RGB}{201,93,99}
\newcommand{\myblue}{blue!80!white!50}
\newcommand{\myred}{red!20!white!90}
\renewcommand{\circ}{\bigcirc}
\renewcommand{\Box}{\square}
\newcommand{\para}[1]{
	\textbf{#1}}

\newcommand{\todogeneric}[1]{{\color{red} #1}}
\newcommand{\todo}[1]{\todogeneric{\textbf{TODO:} #1}}

\newcolumntype{L}{X}
\newcolumntype{R}{>{\raggedleft\arraybackslash}X}
\newcolumntype{C}{>{\centering\arraybackslash}X}

\DeclareGraphicsExtensions{.pdf, .png, .jpg}

\RequirePackage{marvosym}

\DeclarePairedDelimiter{\delimabs}{\lvert}{\rvert}
\DeclarePairedDelimiter{\delimnorm}{\lVert}{\rVert}
\DeclarePairedDelimiter{\delimpospart}{\lgroup}{\rgroup^+}

\DeclarePairedDelimiterX{\deliminner}[2]{\lange}{\rangle}{#1, #2}
\DeclarePairedDelimiter{\delimcardinality}{\lvert}{\rvert}
\DeclarePairedDelimiter{\delimset}{\lbrace}{\rbrace}
\DeclarePairedDelimiter{\delimtuple}{(}{)}
\DeclarePairedDelimiter{\delimlistt}{[}{]}
\DeclarePairedDelimiter{\delimfun}{(}{)}

\NewDocumentCommand{\abs}{sm}{\IfBooleanTF{#1}{\delimabs{#2}}{\delimabs*{#2}}}
\NewDocumentCommand{\norm}{sm}{\IfBooleanTF{#1}{\delimnorm{#2}}{\delimnorm*{#2}}}
\NewDocumentCommand{\pospart}{sm}{\IfBooleanTF{#1}{\delimpospart{#2}}{\delimpospart*{#2}}}
\NewDocumentCommand{\negpart}{sm}{\IfBooleanTF{#1}{\delimnetpart{#2}}{\delimnetpart*{#2}}}
\NewDocumentCommand{\inner}{sm}{\IfBooleanTF{#1}{\deliminner{#2}}{\deliminner*{#2}}}
\NewDocumentCommand{\cardinality}{sm}{\IfBooleanTF{#1}{\delimcardinality{#2}}{\delimcardinality*{#2}}}
\NewDocumentCommand{\set}{sm}{\IfBooleanTF{#1}{\delimset*{#2}}{\delimset{#2}}}
\NewDocumentCommand{\tuple}{sm}{\IfBooleanTF{#1}{\delimtuple{#2}}{\delimtuple*{#2}}}
\NewDocumentCommand{\closure}{sm}{\IfBooleanTF{#1}{\delimclosure{#2}}{\delimclosure*{#2}}}
\NewDocumentCommand{\listt}{sm}{\IfBooleanTF{#1}{\delimlistt{#2}}{\delimlistt*{#2}}}
\NewDocumentCommand{\fun}{smm}{\IfBooleanTF{#1}{{#2}\delimfun{#3}}{{#2}\delimfun*{#3}}}
\NewDocumentCommand{\funMacro}{smm}{\IfNoValueTF{#3}{#1}{\fun{#2}{#3}}}

\DeclareMathOperator{\ExistsOp}{\exists}
\DeclareMathOperator{\ForallOp}{\forall}

\NewDocumentCommand{\Exists}{gg}{\IfNoValueTF{#1}{\ExistsOp}{\ExistsOp #1. \, #2}}
\NewDocumentCommand{\Forall}{gg}{\IfNoValueTF{#1}{\ForallOp}{\ForallOp #1. \, #2}}

\newcommand{\intersectionSym}{\cap}
\newcommand{\intersectionBin}{\mathbin{\intersectionSym}}
\newcommand{\UnionSym}{\bigcup}

\newcommand{\intersection}{\intersectionBin}
\newcommand{\Union}{\UnionSym}

\newcommand{\Naturals}{\mathbb{N}}

\newcommand{\Reals}{\mathbb{R}}

\newcommand{\Distributions}{\mathcal{D}}

\NewDocumentCommand{\convto}{G{}}{\xrightarrow{#1}}
\NewDocumentCommand{\weakto}{G{}}{\xrightharpoonup{#1}}
\NewDocumentCommand{\weakstarto}{G{}}{\xrightharpoonup[*]{#1}}

 \DeclareDocumentCommand{\diff}{D<>{} O{}  D(){}}{\Delta_{#1}^{#2}\ifthenelse{\isempty{#3}}{}{(#3)}}

\DeclareMathOperator*{\argmin}{arg\, min}
\DeclareDocumentCommand{\post}{D<>{} O{} D(){}}{\mathsf{Post}_{#1}^{#2}\ifthenelse{\isempty{#3}}{}{(#3)}}
\DeclareMathOperator{\leaves}{\mathbin{\mathop{\sf exits}}}

\newcommand{\eqdef}{\vcentcolon=}

\newcommand{\reach}{\Diamond}

\NewDocumentCommand{\distributions}{d()}{\funMacro{\mathcal{D}}{#1}}

\newcommand{\bellman}{\mathfrak B}

\newcommand{\MECs}{\mathsf{MEC}}

\DeclareDocumentCommand{\val}{D<>{} O{}  D(){} t'}{\achievable_{#1}^{\IfBooleanTF{#4}{\prime #2}{#2}}\ifthenelse{\isempty{#3}}{}{(#3)}}
\DeclareDocumentCommand{\ub}{D<>{} O{}  D(){} t'}{\mathsf{U}_{#1}^{\IfBooleanTF{#4}{\prime #2}{#2}}\ifthenelse{\isempty{#3}}{}{(#3)}}
\DeclareDocumentCommand{\gub}{D<>{} O{}  D(){} t'}{\mathsf{G}_{#1}^{\IfBooleanTF{#4}{\prime #2}{#2}}\ifthenelse{\isempty{#3}}{}{(#3)}}
\DeclareDocumentCommand{\lb}{D<>{} O{}  D(){} t'}{\mathsf{L}_{#1}^{\IfBooleanTF{#4}{\prime #2}{#2}}\ifthenelse{\isempty{#3}}{}{(#3)}}
\DeclareDocumentCommand{\game}{D<>{} O{} D(){} t'}{\mathsf{G}_{#1}^{\IfBooleanTF{#4}{\prime}{}#2}\ifthenelse{\isempty{#3}}{}{(#3)}}
\DeclareDocumentCommand{\transition}{D<>{} O{} D(){}}{\rightarrow_{#1}^{#2}\ifthenelse{\isempty{#3}}{}{(#3)}}

\newcommand{\SG}{\textrm{SG}}
\newcommand{\SGs}{\textrm{SGs}}
\DeclareDocumentCommand{\G}{D<>{} O{} t' D(){}}{\mathsf{G}_{#1}^{\IfBooleanTF{#3}{\prime}{}#2}\ifthenelse{\isempty{#4}}{}{(#4)}}
\DeclareDocumentCommand{\exGame}{D<>{} O{} t'  D(){}}{\mathsf{G}_{#1}^{\IfBooleanTF{#3}{\prime}{}#2}\ifthenelse{\isempty{#4}}{}{(#4)}=(\states<#1>[\IfBooleanTF{#3}{\prime}{}#2],\states<\Box\ifthenelse{\isempty{#1}}{}{,#1}>[\IfBooleanTF{#3}{\prime}{}#2],\states<\circ\ifthenelse{\isempty{#1}}{}{,#1}>[\IfBooleanTF{#3}{\prime}{}#2],\istate<#1>[\IfBooleanTF{#3}{\prime}{}#2],\actions<#1>[\IfBooleanTF{#3}{\prime}{}#2],\Av<#1>[\IfBooleanTF{#3}{\prime}{}#2],\trans<#1>[\IfBooleanTF{#3}{\prime}{}#2])}

\DeclareDocumentCommand{\states}{D<>{} O{}  t'}{\mathit{S}_{#1}^{\IfBooleanTF{#3}{\prime~#2}{#2}}}
\DeclareDocumentCommand{\state}{D<>{} O{}  t'}{\mathsf{s}_{#1}^{\IfBooleanTF{#3}{\prime #2}{#2}}}
\DeclareDocumentCommand{\istate}{D<>{} O{} t'}{\mathsf{s}_{0\ifthenelse{\isempty{#1}}{}{,#1}}^{\IfBooleanTF{#3}{\prime~#2}{#2}}}

\newcommand{\initstate}{\state<0>}

\DeclareDocumentCommand{\trans}{D<>{} O{} t' D(){} D(){}}{\delta_{#1}^{\IfBooleanTF{#3}{\prime}{}#2}\ifthenelse{\isempty{#4}}{}{(#4)}\ifthenelse{\isempty{#5}}{}{(#5)}}
\DeclareDocumentCommand{\Av}{D<>{} O{} t' D(){}}{\mathsf{Av}_{#1}^{\IfBooleanTF{#3}{\prime}{}#2}\ifthenelse{\isempty{#4}}{}{(#4)}}

\DeclareDocumentCommand{\F}{D<>{} O{} t' D(){}}{\mathsf{F}_{#1}^{\IfBooleanTF{#3}{\prime}{}#2}\ifthenelse{\isempty{#4}}{}{(#4)}}

\newcommand{\Path}{\rho}
\DeclareDocumentCommand{\Path}{D<>{} O{} t' D(){}}{\path<#1>[#2]\IfBooleanTF{#3}{'}{}(#4)}
\DeclareDocumentCommand{\path}{D<>{} O{} t' D(){}}{\rho_{#1}^{\IfBooleanTF{#3}{\prime}{}#2}\ifthenelse{\isempty{#4}}{}{(#4)}}
\newcommand{\fpath}{\mathsf{w}}
\DeclareDocumentCommand{\Paths}{D<>{} O{} t' D(){}}{\Omega_{#1}^{\IfBooleanTF{#3}{\prime}{}#2}\ifthenelse{\isempty{#4}}{}{(#4)}}
\newcommand{\straa}{\sigma}

\newcommand{\strab}{\tau}

\DeclareDocumentCommand{\strategy}{D<>{} O{} D(){}
  t*}{{\IfBooleanTF{#4}{\tau}{\sigma}}_{#1}^{#2}\ifthenelse{\isempty{#3}}{}{(#3)}}

\DeclareDocumentCommand{\actions}{D<>{} O{} t' d()}{{\IfNoValueTF{#4}{\mathsf{A}}{\fun{\mathsf{A}}{#4}}}_{#1}^{\IfBooleanTF{#3}{\prime~#2}{#2}}}
\DeclareDocumentCommand{\action}{D<>{} O{} t'}{\mathsf{a}_{#1}^{\IfBooleanTF{#3}{\prime#2}{#2}}}

\newcommand{\mec}{\mathsf{MEC}}
\newcommand{\scc}{\mathsf{SCC}}

\newcommand{\EC} {\mathsf{EC}}
\newcommand{\bscc}{\mathsf{BSCC}}

\newcommand{\pr}{\mathbb P}

\DeclareDocumentCommand{\target}{D<>{} O{}
  t'}{\mathfrak{1}_{#1}^{\IfBooleanTF{#3}{\prime}{}#2}}
\newcommand{\targetset}{T}
\newcommand{\targetsets}{\mathcal T}
\newcommand{\pareto}{\mathfrak P}
\newcommand{\achievable}{\mathfrak A}
\DeclareDocumentCommand{\fail}{D<>{} O{} t'}{\mathsf{\bot}_{#1}^{\IfBooleanTF{#3}{\prime}{}#2}}

\newcommand{\UPDATE}{\mathsf{UPDATE}}

\DeclareDocumentCommand{\CEC}{D<>{} O{} D(){}
  t*}{\IfBooleanTF{#4}{\simple~}{}\ifthenelse{\isempty{#3}}{}{#3\textit{-}}\textrm{BEC}_{#1}^{#2}}
\newcommand{\simple}{simple}

\NewDocumentCommand{\exit}{D<>{}D[]{}D(){}}{\mathsf{BE}_{#1}^{#2}\ifthenelse{\isempty{#3}}{}{(#3)}}

\newcommand{\MSEC}{MSEC}
\newcommand{\SEC}{SEC}

\newcommand{\cl}{\mathsf{clos}}

\newcommand{\drawcirc}{\node[draw,circle,minimum size=.7cm, outer sep=1pt]}
\newcommand{\drawbox}{\node[draw,rectangle,minimum size=.7cm, outer sep=1pt]}
\newcommand{\drawdummy}{\node[minimum size=0,inner sep=0]}
\newcommand{\drawnum}{\node[minimum size=.4cm,inner sep=1pt]}

\newcommand{\sink}{\mathfrak 0}

\newcommand{\directions}{\mathsf D}
\newcommand{\conv}{conv}
\newcommand{\dwc}{\mathit{dwc}}
\newcommand{\deflate}{\mathsf{DEFLATE}}
\newcommand{\DEFLATE}{\deflate}

\newcommand{\dif}{c}
\newcommand{\region}{region}
\newcommand{\regions}{\mathcal{R}}
\newcommand{\rSEC}{\SEC}
\newcommand{\dir}{\mathbf d}
\newcommand{\exits}{\mathsf{Exits}}
\newcommand{\doSECs}{\mathsf{DEFLATE\_\SEC{}s}}
\newcommand{\FINDSECS}{\mathsf{FIND\_\SEC{}s}}
\newcommand{\GETREGIONS}{\mathsf{GET\_REGIONS}}

\newcommand{\mop}{\overline{\bellman}}

\newcommand{\consistent}{consistent}

\acmConference[LICS '20]{Proceedings of the 35th Annual ACM/IEEE Symposium on Logic in Computer Science}{July 8--11, 2020}{Saarbr\"ucken, Germany}
\acmBooktitle{Proceedings of the 35th Annual ACM/IEEE Symposium on Logic in Computer Science (LICS '20), July 8--11, 2020, Saarbr\"ucken, Germany}
\acmYear{2020}
\copyrightyear{2020}
\setcopyright{rightsretained}
\acmISBN{978-1-4503-7104-9/20/07}
\acmDOI{10.1145/3373718.3394761}
\startPage{1}
\advance\textwidth4mm 
\advance\hoffset-2mm

\usepackage{etoolbox}
\newtoggle{arxiv}
\newcommand{\ifarxivelse}[2]{\iftoggle{arxiv}{#1}{#2}}
\toggletrue{arxiv} 

\begin{document}

\title{Approximating Values of Gen\-er\-al\-ized-Reachability Stochastic Games}

\author{Pranav Ashok}
\affiliation{          
	\institution{Technical University of Munich}            
	\country{Germany}                    
}
\email{ashok@in.tum.de}

\author{Krishnendu Chatterjee}
\affiliation{           
	\institution{IST Austria}            
	\country{Austria}                    
}
\email{Krishnendu.Chatterjee@ist.ac.at}

\author{Jan K\v{r}et{\'i}nsk{\'y}}
\affiliation{           
	\institution{Technical University of Munich}            
	\country{Germany}                    
}
\email{jan.kretinsky@in.tum.de}

\author{Maximilian Weininger}
\affiliation{             
	\institution{Technical University of Munich}            
	\country{Germany}                   
}
\email{maxi.weininger@tum.de}

\author{Tobias Winkler}
\affiliation{             
	\institution{RWTH Aachen University}            
	\country{Germany}                   
}
\email{tobias.winkler@cs.rwth-aachen.de}

\renewcommand{\shortauthors}{P. Ashok, K. Chatterjee, J. K\v{r}et{\'i}nsk{\'y}, M. Weininger, T. Winkler} 



\begin{abstract}
  Simple stochastic games are turn-based 2\textonehalf-player games with a reachability objective. The basic question asks whether one player can ensure reaching a given target with at least a given probability. A natural extension is games with a conjunction of such conditions as objective. Despite a plethora of recent results on the analysis of systems with multiple objectives, the decidability of this basic problem remains open. In this paper, we present an algorithm approximating the Pareto frontier of the achievable values to a given precision. Moreover, it is an anytime algorithm, meaning it can be stopped at any time returning the current approximation and its error bound.
\end{abstract} 

\ccsdesc[300]{Theory of computation~Algorithmic game theory}
\ccsdesc[500]{Theory of computation~Verification by model checking}
\ccsdesc[300]{Mathematics of computing~Probability and statistics}

\keywords{Stochastic games; Multiple Reachability Objectives; Pareto frontier; Anytime algorithm}

\maketitle

\section{Introduction} \label{sec:intro}

\para{Simple stochastic games} \cite{condonComplexity} are zero-sum turn-based stochastic games (SG) with two players, which we call Maximizer and Minimizer. The objective of player Maximizer is to maximize the probability of reaching a given target set of states, while player Minimizer aims at the opposite.
The basic decision problem is to determine whether there is a strategy for Maximizer achieving at least a given probability threshold.
These games are interesting theoretically: the problem is known to be in NP$\,\cap\,$co-NP, but whether it belongs to P is a major and long-standing problem. 
Moreover, several other important game problems such as parity games reduce to it \cite{DBLP:journals/corr/abs-1106-1232}.
Besides, they are also practically relevant: they can serve as a tool for synthesis with safety/co-safety objectives in environments with stochastic uncertainty.

\para{Multi-objective stochastic systems} have attracted a lot of attention recently, both SG and the special case with only one player (Markov decision processes, MDP \cite{Puterman}).
They model and enable optimization with respect to conflicting goals, where a desired trade-off is to be considered. 
A natural multi-dimensional generalization of the reachability threshold constraint $\pr[\reach T]\geq t$ is a conjunction $\bigwedge_i\pr[\reach T_i]\geq t_i$ giving rise to \emph{generalized-reachability} (or \emph{multiple-reachability}) \emph{stochastic games}, similar to e.g. generalized mean-payoff SG \cite{DBLP:conf/tacas/BassetKTW15,DBLP:conf/lics/Chatterjee016}. 
The problem is then to decide whether a given vector of thresholds can be achieved by Maximizer.
Note that these games are not determined \cite{DBLP:conf/mfcs/ChenFKSW13}, and in this paper we consider the lower-value (worst-case) problem formulation, i.e. finding a strategy of Maximizer that can guarantee the vector no matter what Minimizer does.

The main results established in the literature are as follows.
For MDP, while generalized mean-payoff can be solved in P \cite{DBLP:conf/fsttcs/Chatterjee07a,DBLP:journals/corr/abs-1104-3489}, generalized reachability is PSPACE-hard and can be solved in exponential time \cite{DBLP:conf/cav/RandourRS15}. 
For SG, generalized mean-payoff has been solved for almost-sure conditions only \cite{DBLP:conf/tacas/BassetKTW15,DBLP:conf/lics/Chatterjee016} and approximation of the values for generalized mean-payoff as well as generalized reachability are still open.
The generalized-reachability SG problem is only known to be decidable for the subclass of \emph{stopping} SGs with a \emph{2-dimensional} objective \cite{DBLP:conf/atva/BrenguierF16}
(an SG is stopping if under any strategies a designated set of sinks is reached almost surely). 

The main open question for genera\-lized-reachability SG is decidability.
There are several important subgoals towards this problem:
From the decidability perspective, stopping SG with more than 2-dimensional objectives, or general SG with 2- (or more) dimensional objectives have been open.
Moreover, the same holds even for $\varepsilon$-approximability.
From the algorithmic perspective,  \cite{DBLP:conf/mfcs/ChenFKSW13} provides a converging sequence of \emph{lower bounds} on the \emph{Pareto frontier}, i.e. the set of achievable vectors that are pointwise-maximal (in other words, vectors which cannot be improved in one dimension without sacrificing another one). 
It is open whether converging upper bounds can also be computed.
Since such bounds would imply $\varepsilon$-approximability, this open question is the most imminent.

\para{Our contribution} in this paper is twofold. Firstly, we prove the following theorem: 
\begin{quote}
\para{Theorem:}\emph{
	The set of all achievable vectors in an \emph{arbitrary} (not necessarily stopping) SG with generalized-reachability objective of \emph{any dimension} can be effectively approximated for any given precision $\varepsilon$.
}
\end{quote}
Secondly, we provide a value-iteration algorithm that approximates the Pareto frontier by giving converging lower \emph{and upper} bounds.
Consequently, it becomes an \emph{anytime algorithm}, providing the current approximation and its error at each moment of the computation.
Thus our first contribution resolves the $\varepsilon$-approximability open question and our second contribution resolves the algorithmic open question; both results are for arbitrary SG with generalized-reachability objectives of any dimension.

\para{Convergent upper bounds on the value} are known to be notoriously difficult to achieve.
Until recently, the default engine for analysis in the most used probabilistic model checker PRISM \cite{prism} and PRISM-Games \cite{PRISM-games} used \emph{value iteration}, e.g.\ \cite{Puterman}, which converges to the value from below, but because of the used stopping criteria the results could be arbitrarily wrong~\cite{BVI}.
For a solution with a given precision, one could use linear programming instead, which however, does not scale well for MDP, and, more importantly, does not work at all for SG~\cite{condonAlgSSG}.
For MDP, value iteration has been extended~\cite{atva,BVI} so that it provides not only the under-approximating convergent sequence, but also an over-approxi\-mating one, calling the technique ``bounded value iteration'' (due to \cite{BRTDP}) or ``interval iteration'', respectively.
Its essence is to collapse \emph{maximal end components} (MECs) of the MDP, thereby not changing the values; on MDP without MECs the over-approximating sequence converges to the actual value of the (collapsed as well as original) MDP.
This technique was further extended to MDP with mean-payoff objective \cite{cav17}.
In contrast, for SG one cannot collapse MECs since they account for non-trivial alternating structure, as opposed to MDP, where any desired action exiting the MEC can be taken almost surely.
Therefore, a more complex procedure has been proposed for SG \cite{cav18}: Depending on the current under-approximation, problematic parts of MECs are dynamically identified and their over-approximation is lowered to over-approximations of certain actions exiting the MEC, as exemplified and explained later. 
We lift this procedure to general dimensions.
Note that we do not give any convergence rate for our algorithm, because it is not possible to extend the argument of the single-dimensional case~\cite{visurvey} in a straightforward manner. This argument requires the lowest probability occurring in a play to be bounded. However, in the multi-dimensional setting strategies may need infinite memory and hence there is no lower bound on the probability that a strategy assigns to actions \cite[Appendix B1, full version]{cfk13}. Giving bounds on the convergence is an interesting direction of future work.

This paper combines and extends several techniques from literature to obtain the corresponding result for the multi-dimensional case:
\begin{itemize}
	\item Firstly, we use the Bellman operator extended to down-ward-closed sets (instead of just real values) \cite{DBLP:conf/mfcs/ChenFKSW13}, allowing for value iteration in the multi-dimensional setting.
	\item Secondly, we exploit the technique of \cite{cav18}, which in the single-dimensional setting repetitively identifies the problematic parts of MECs hindering convergence.
	\item Thirdly, in order to apply this technique, we reduce the multi-dimensional problem to a continuum of single-dimensional problems, by splitting the Pareto front into directions, similar to \cite{DBLP:conf/atva/ForejtKP12}.
	\item Fourthly, we group the single-dimensional problems into finitely many \emph{regions}, similar in spirit to regions of timed automata \cite{DBLP:journals/tcs/AlurD94} since they are essentially given by orderings of the approximate values of certain actions.
	Nevertheless, due to the projective geometry of the problem, we need to work slightly more generally with simplicial complexes, see e.g.\ \cite{hatcher2002algebraic}.
\end{itemize}
The main technical difficulty is to identify (i) the parts of MECs with an unjustified too high upper bound and (ii) the value to which it should be decreased in each step.
Both of these depend on the desired trade-off between the targets.
As we compute the whole set of achievable vectors, we need to consider all possible trade-offs, which are, moreover, uncountably many.

\smallskip

\para{Related work.}
Already for a decade, MDP have been extensively studied in the setting of multiple objectives.
Multiple objectives have been considered both qualitative, such as reachability and LTL~\cite{DBLP:journals/lmcs/EtessamiKVY08}, as well as quantitative, such as mean payoff~\cite{DBLP:conf/fsttcs/Chatterjee07a,DBLP:journals/corr/abs-1104-3489}, discounted sum~\cite{DBLP:conf/lpar/ChatterjeeFW13}, or total reward~\cite{DBLP:conf/tacas/ForejtKNPQ11}.
The expectation has been combined with variance in~\cite{DBLP:conf/lics/BrazdilCFK13}.
Beside expectation queries, conjunctions of percentile (threshold) queries have been considered for various objectives \cite{FKR95,DBLP:journals/corr/abs-1104-3489,DBLP:journals/fmsd/RandourRS17,DBLP:journals/lmcs/ChatterjeeKK17}.
Further, for general Boolean combinations
for Markov chains with total reward, \cite{DBLP:conf/lics/HaaseKL17} approximates the value, while computability is still open.
In contrast, \cite{DBLP:conf/fossacs/Velner15} shows that Boolean combinations over mean payoff games become quickly undecidable.
For the specifics of the two-dimensional case and the interplay of the two objectives, see \cite{DBLP:conf/nfm/BaierDDKK14}.
The usage of the multi-dimensional setting is discussed in \cite{DBLP:conf/fase/BaierDKDKMW14,DBLP:conf/csl/BaierDK14}, comparing multiple rewards and quantiles and reporting how they have practically been applied and found useful by domain experts.

More recently, SG have been also analyzed with multiple objectives; \cite{DBLP:journals/ejcon/SvorenovaK16} provides an overview and implementation of existing algorithms for Pareto frontier computation for multi-objective total reward, reachability, and probabilistic LTL properties as well as mixtures thereof. However, the computation is limited to stopping SGs, i.e. ones without end components

Multiple mean-payoff objective was first examined in \cite{DBLP:conf/tacas/BassetKTW15} and both the qualitative and the quantitative problems are coNP-complete \cite{DBLP:conf/lics/Chatterjee016}.
Although Boolean combinations of mean-payoff are undecidable in general \cite{DBLP:conf/fossacs/Velner15}, in certain subclasses of SG they can be approximated \cite{DBLP:journals/iandc/BassetKW18}.
Boolean  combinations  of  total-reward  objectives  were  approximated in the case of stopping games \cite{DBLP:conf/mfcs/ChenFKSW13} and applied to autonomous driving \cite{DBLP:conf/qest/ChenKSW13}, where LTL is reduced to total reward in the case of stopping games and,
for dimension two, the problem is shown decidable in \cite{DBLP:conf/atva/BrenguierF16}.

PRISM-Games \cite{DBLP:conf/tacas/KwiatkowskaPW16} provides tool support for several multi-player multi-objective settings \cite{DBLP:journals/sttt/KwiatkowskaPW18}. Other tools supporting multi-player settings, GAVS+ \cite{DBLP:conf/tacas/ChengKLB11} and GIST \cite{DBLP:conf/cav/ChatterjeeHJR10}, are not maintained any more and are limited to single-objective settings. 

In many settings, Pareto frontiers can be $\varepsilon$-approximated in polynomial time \cite{DBLP:conf/focs/PapadimitriouY00}.
Pareto frontiers are constructed for the generalized mean-payoff objective for 2-player (non-stochastic) games in \cite{DBLP:conf/cav/BrenguierR15}, MDPs in \cite{DBLP:journals/corr/abs-1104-3489,DBLP:journals/lmcs/ChatterjeeKK17}, and SGs in \cite{DBLP:journals/iandc/BassetKW18}.
For the genera\-lized-reachability, the Pareto frontier is approximated for MDP in \cite{DBLP:journals/lmcs/EtessamiKVY08}, but for SG the Pareto frontier is not even known to be given by finitely many points, except for dimension two \cite{DBLP:conf/atva/BrenguierF16}.
In contrast, in the single-dimensional case, the value is known to be a multiple of a denominator that can be calculated from the syntactic description of the game \cite{visurvey}.

%
%
%

\para{Structure of the paper} 
After recalling the basic notions in Section~\ref{sec:prelim}, we illustrate the problem, the difficulties and our solution on examples in Section \ref{sec:example}.
The algorithm is described and the correctness intuitively explained in Section~\ref{sec:algo} and formally proven in Section~\ref{sec:proofStruct}.
The proofs of several technical statements are, for the sake of readability, relayed to \ifarxivelse{Appendix}{\cite[Appendix]{techreport}}.
We conclude in Section~\ref{sec:conc}.

\section{Preliminaries} \label{sec:prelim}
\subsection{Stochastic Games}

A 
{probability distribution} on a finite set $X$ is a mapping $\trans: X \to [0,1]$, such that $\sum_{x\in X} \trans(x) = 1$.
The set of all probability distributions on $X$ is denoted by $\Distributions(X)$.
%
Given a dimension $n\in \Naturals$, often implicitly clear from context, and $c\in \Reals$, we let $\vec c$ denote the $n$-dimensional vector with all components equal to $c$.
For a vector $\vec v$, its $i$-th component is denoted $\vec v_i$.
We compare vectors component-wise, i.e. $\vec u\leq \vec v$ if $\vec u_i\leq \vec v_i$ for all $i$.
%
In this paper, we restrict ourselves to non-negative vectors, i.e.\ elements of $\Reals_{\geq0}^n$

Now we define turn-based two-player stochastic games.
As opposed to the notation of e.g.\ \cite{condonComplexity}, we do not have special stochastic nodes, but rather a probabilistic transition function.
 
\begin{definition}[\SG]
	A \emph{stochastic game ($\SG$)} is a tuple 
	$(\states, \states<\Box>,\allowbreak \states<\circ>, \initstate, \actions, \Av, \delta)$,
	where $\states$ is a finite set of \emph{states} partitioned
	\ into the sets $\states<\Box>$ and $\states<\circ>$ of states of the player \emph{Maximizer} and \emph{Minimizer}, 
	respectively, 
	$\initstate \in \states$ is the \emph{initial} state, $\actions$ is a finite set of \emph{actions}, $\Av: \states \to 2^{\actions}$ assigns to every state a set of \emph{available} actions, and $\trans: \states \times \actions \to \distributions(\states)$ is a \emph{transition function} that given a state $\state$ and an action $\action\in \Av(\state)$ yields a probability distribution over \emph{successor} states.	
\end{definition}
A {Markov decision process (MDP)} is then a special case of $\SG$ where $\states<\circ> = \emptyset$.
We assume that $\SG$ are non-blocking, so for all states $\state$ we have $\Av(\state) \neq \emptyset$.

For a state $\state$ and an available action $\action \in \Av(\state)$, we denote the set of successors by $\post(\state,\action) := \set{\state' \mid \trans(\state,\action,\state') > 0}$. 
We say a state-action pair $(\state, \action)$ is an \emph{exit} of a set of states $T$, written $(\state,\action)\leaves T$, if $\exists t \in \post(\state,\action):t \notin T$, i.e., if with some probability a successor outside of $T$ could be chosen. 
Further, we use $\exits(T) = \set{(\state,\action) \mid \state \in T, \action \in \Av(\state), (\state,\action)\leaves T}$ to denote all exits of a state set $T \subseteq \states$.
Finally, for any set of states $T \subseteq \states$, we use $T_\Box$ and $T_\circ$ to denote the states of $T$ that belong to Maximizer and Minimizer, whose states are drawn in the figures as $\square$ and $\circ$, respectively.
 

The semantics of SG is given in the usual way by means of strategies and the induced Markov chain \cite{BaierBook} and its respective probability space, as follows.
An \emph{infinite path} $\path$ is an infinite sequence $\path = \state<0> \action<0> \state<1> \action<1> \cdots \in (\states \times \actions)^\omega$, such that for every $i \in \Naturals$, $\action<i>\in \Av(\state<i>)$ and $\state<i+1> \in \post(\state<i>,\action<i>)$.
\emph{Finite path}s are defined analogously as elements of $(\states \times \actions)^\ast \times \states$.
%
A \emph{strategy} of Maximizer or Minimizer is a function $\straa: (\states \times \actions)^\ast \times \states<\Box> \to \distributions(\actions)$ or $(\states \times \actions)^\ast \times \states<\circ> \to \distributions(\actions)$, respectively, such that $\straa(\rho \state) \in \distributions(\Av(\state))$ for all $\state$.
We call a strategy \emph{deterministic} if it maps to Dirac distributions only; otherwise, it is \emph{randomizing}.
A pair $(\straa,\strab)$ of strategies of Maximizer and Minimizer induces an (infinite state) Markov chain $\G[\straa,\strab]$ with finite paths as states, $\initstate$ being initial, and the transition function $\trans(\fpath \state, \fpath \state\action \state') = \straa(\fpath \state) (\action) \cdot \trans(\state, \action, \state')$ for states of Maximizer and analogously for states of Minimizer, with $\straa$ replaced by $\strab$.
The Markov chain 
induces a unique probability distribution $\pr^{\straa,\strab}$ over measurable sets of infinite paths \cite[Ch.~10]{BaierBook} (the usual index with the initial state is not used since it is fixed already in the game).

\subsection{End Components}
Now we recall a fundamental tool for analysis of MDP called end components.
An end component of a SG is then defined as the end component of the underlying MDP with both players unified. 
\begin{definition}[EC]
	\label{def:EC}
	A non-empty set $T\subseteq \states$ of states is an \emph{end component (EC)} if there is a non-empty set $B \subseteq \Union_{\state \in T} \Av(s)$ of actions such that 
	\begin{enumerate}
		\item for each $\state \in T, \action \in B \intersection \Av(\state)$, we have $(\state,\action)\notin \exits(T)$, 
		\item for each $\state, \state' \in T$ there is a finite path $\fpath = \state \action<0> \dots \action<n> \state' \in (T \times B)^* \times T$, i.e. the path stays inside $T$ and only uses actions in $B$.
	\end{enumerate}
\end{definition}

Intuitively, ECs correspond to bottom strongly connected components of the Markov chains induced by possible strategies.
Hence for some pair of strategies all possible paths starting in an EC remain there. 
An EC $T$ is a \emph{maximal end component (MEC)} if there is no other end component $T'$ such that $T \subseteq T'$.
Given an $\SG$ $\G$, the set of its MECs is denoted by $\mec(\G)$ and can be computed in polynomial time~\cite{CY95}.

\subsection{Generalized Reachability}

For a set $\targetset\subseteq\states$, we write $\Diamond \targetset:=\set{\state<0> \action<0> \state<1> \action<1> \cdots \in(\states \times \actions)^\omega  \mid \exists i \in \Naturals: \state<i>\in\targetset}$ to denote the (measurable) set of all paths which eventually reach $\targetset$. 
A \emph{generalized-reachability objective} (of dimension $n$) is an $n$-tuple $\targetsets=(\targetset_1,\ldots,\targetset_n)$ of state sets $\targetset_i\subseteq \states$.
A vector $\vec v$ (of dimension $n$) is \emph{achievable} if there is a strategy $\sigma$ of Maximizer such that for all strategies $\tau$ of Minimizer 
\[\forall i\in\{1,\ldots,n\}\quad\pr^{\straa,\strab}(\Diamond \targetset_i)\geq \vec v_i\]
Note that, since these games are not determined \cite{DBLP:conf/mfcs/ChenFKSW13}, this corresponds to the lower value, i.e. the worst case analysis.

For a given state $\state$, the set of points achievable \emph{from} $\state$, meaning in a game where the initial state is set to $\state$, is denoted $\achievable_{\targetsets}(s)$ or just $\achievable(s)$ when $\targetsets$ is clear from context.


\subsection{Basic Geometry Notation and Pareto Frontiers}\label{sec:prelimGeometry}

In order to consider convex combinations of sets, we define scaling of a set $X \subseteq \Reals^n$ by a constant $c \in [0,1]$ as
$c \cdot X = \set{c \cdot x \mid x \in X}$, 
and the Minkowski sum of sets $X$ and $Y$ as
$X + Y = \set{x+y \mid x \in X, y \in Y}$.
The convex hull of a set $X$ is denoted by $\conv (X)=\{\sum_{i=1}^k a_ix_i \mid k\geq 1,  x_i\in X, a_i\geq 0, \sum_{i=1}^k a_i=1\}.$ 

A \emph{downward closure} of a set $X$ of vectors is $\dwc(X) := \{y \mid \exists x \in X: y \leq x\}$. 
A set $X$ is \emph{downward closed} if $X=\dwc(X)$.
The set $\achievable$ of achievable points is clearly downward closed.

It will be convenient to use a few basic notions of projective geometry, which we now recall.
Intuitively, a \emph{direction} is a ray from the origin $\vec 0$ into the ($n$-dimensional) first quadrant.
As such, we may represent it with any vector $\vec v \neq \vec 0$ on that ray.
Then all vectors $\lambda\cdot \vec v$ for any $\lambda\in\Reals_{>0}$ are equivalent and represent the same direction. 
For instance, direction $\dir=[(1,0,0)]$ denotes the $x$-axis and it is equal to  $[(\lambda,0,0)]$ for any $\lambda>0$.
Formally, a direction $\dir=[\vec v]$ is the set $\{\lambda\cdot\vec v\mid \lambda\in\Reals_{>0}\}$.
We denote by $\directions=\{[\vec v]\mid \vec v\in\dwc(\{\vec 1\})\}$ the set of all directions (in the first quadrant).


Given a set $X$ of points and a direction $\dir$, $X$ \emph{evaluated in direction} $\dir$ is the (Euclidean) length of the vector from the origin to the farthermost intersection of $X$ and $\dir$, denoted $$X[\dir]:=\sup\{||\vec x|| \mid \vec x\in X,\dir=[\vec x]\}$$ with the usual $\sup\emptyset=0$.
Fig.~\ref{fig:intersection-with-Pareto-frontier} illustrates an evaluation of a direction on an achievable set.
Intuitively, it describes what is achievable if we prefer the dimensions in the ``ratio'' given by $\dir$.
Another example is the whole set (blue and red) of Fig.~\ref{fig:vis-regions}: evaluated in $[(1,1)]$ it yields $\sqrt2/2$.

Given a downward closed set $X$, its \emph{Pareto frontier} is the set of farthermost points in each direction:
\[\pareto(X)=\{\vec x  \mid \dir\in\directions, \dir=[\vec x], X[\dir]=||\vec x||\}\]
The \emph{Pareto frontier of a state $\state$} is the Pareto frontier of the set achievable in $\state$, i.e. $\pareto(\state):=\pareto(\achievable(\state))$.
The \emph{Pareto set  of the game} is $\pareto:=\pareto(\initstate)$.
Thus by definition, $\pareto=\pareto(\achievable(\initstate))$ and, further, $\dwc(\pareto)$ is (the closure of) $\achievable(\initstate)$.%
\footnote{Our notion of Pareto frontier captures the whole surface in the first quadrant.
	Other definitions such as $\pareto_{\targetsets}=\{\vec v\mid \vec v\text{ is achievable }\wedge\forall\text{ achievable }\vec u: \vec u\not>\vec v\}$ only capture the Pareto optimal points.
	For example, if the set of achievable points in the three-dimensional space is the whole unit cube then our definition returns its three sides, while the other definition returns only the singleton with the Pareto optimal point $(1,1,1)$.}
Note that it is not known whether $\val$ is closed, since it is not known whether the suprema of achievable points are also achievable. 
Our notion of $\pareto$ includes these suprema, which is why it is only equal to the closure of $\achievable$.

\subsection{Problem Formulation}\label{sec:problem}

In this paper, we are interested in $\varepsilon$-approximating $\pareto$. 
In terms of under- and over-approximation:

 \medskip
\noindent\framebox[\linewidth]{\parbox{0.95\linewidth}{
Given an SG, generalized-reachability objective $\targetsets$, and precision $\varepsilon>0$, 
the task is to construct sets $\mathcal L,\mathcal U\subseteq \Reals^{|\targetsets|}$ such that for each direction $\dir\in \directions$, 
$\mathcal L[\dir]$ and $\mathcal U[\dir]$ are effectively computable and we have
\[\mathcal L[\dir]\leq \pareto[\dir]\leq\mathcal  U[\dir] \quad\text{ and }\quad \mathcal U[\dir]-\mathcal L[\dir]<\varepsilon\ .\]
}}

\begin{figure}
	\centering\begin{tikzpicture}[scale=1.4, font=\footnotesize]
	\draw[draw=none,-] (-0.4, -0.4) -- (-0.4, 1.2) -- (1.2, 1.2) -- (1.2, -0.4) -- (-0.4, -0.4);
	
	\draw[->]  (0, -0.1) -- (0, 1.2);
	\draw[->]  (-0.1, 0) -- (1.2, 0);
	
	\draw[-]   (0, 1) -- (0.5, 0.9) -- (0.9, 0.5) -- (1, 0);
	
	\draw[->]  (0, 0) -- (1.2, 0.35);
	
	\draw[draw=none,fill=red,radius=1pt] (0.95, 0.275) circle node [right]{~~~$\dir$} ;
	\node at(0.4,0.5) {$X$} ;
	
	\draw[fill=black,radius=0.75pt] (0.5, 0.9) circle;
	\draw[fill=black,radius=0.75pt] (0.9, 0.5) circle;
	\draw[fill=black,radius=0.75pt] (0, 1) circle node [left] {1};
	\draw[fill=black,radius=0.75pt] (1, 0) circle node [below] {1};
	\end{tikzpicture}
	\caption{Example showing a Pareto frontier of a set $X$, a direction $\dir$, and the point of intersection of $\dir$ with the frontier, depicted as  \protect\tikz\protect\draw[draw=none,fill=red] (0,0) circle (.5ex); in distance $X[\dir]$ from the origin.}
	\label{fig:intersection-with-Pareto-frontier}
\end{figure}
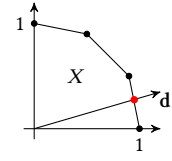

\subsection{Multi-dimensional and Bounded Value Iteration}\label{sec:bvi}

In this section we recall two extensions of the standard value iteration: a generalization for \emph{multi-dimensional} objectives and a \emph{``bounded''} one with an over-approximating sequence.
Firstly, the multi-dimensional Bellman operator for reachability, e.g.\ \cite{DBLP:conf/mfcs/ChenFKSW13},
\[\bellman:\big(\states\to 2^{[0,1]^n}\big)\to\big(\states\to 2^{[0,1]^n}\big)\]  
works with \emph{sets} $X(\state)$ of points achievable in $\state$ rather than single points:

\begin{equation*}\label{eq:s}
\bellman(X)(\state) = \begin{cases}
~~~~~~~~~\bigcap_{a \in \Av(\state)} X(\state,\action) &\mbox{if } \state \in \states<\circ>\\
\conv(\bigcup_{a \in \Av(\state)} X(\state,\action)) &\mbox{if } \state \in \states<\Box> 
\end{cases}
\end{equation*}
where we define
\begin{equation*}\label{eq:sa}
X(\state,\action) =\left( \dwc(\{\mathbbm 1_{\targetsets}(\state)\}) + \sum_{\state' \in \states 
} \trans(\state,\action,\state') \cdot X(\state')\right) \cap \mathbf 1
\end{equation*}
and $\mathbbm 1_{\targetsets}$ is the indicator vector function of target sets, i.e. 
$\mathbbm 1_{\targetsets}(\state)_i$ equals $1$ if $\state \in \targetset_i$ and $0$ otherwise, 
and $\mathbf 1=\dwc(\{\vec 1\})$ is the unit box.

Intuitively, the operator works as follows.
Given what can be achieved from $\state$ using now an action $\action$, we can compute the value for the minimizing state as the intersection over all actions since these points are achievable no matter what Minimizer does.
For maximizing states, if there exists an action achieving a point then Maximizer can achieve it from here; moreover, we compute the convex hull since Maximizer can also randomize and, as opposed to the minimizing case with intersection, union of convex sets need not be convex.
Once we have dealt with decision making on the first line, it remains to determine what can be achieved by each decision, on the second line.
The achievable values are given by the weighted average of the successors' values, but additionally, the base case of targets must be handled.
Namely, whenever a state is in a target set, all values up to $1$ in that dimension are achievable (but not greater than $1$).

This also gives rise to an algorithm approximating $\val$, which is the least fixpoint of $\bellman$~\cite{DBLP:conf/mfcs/ChenFKSW13}.
We initialize $\lb:\states\to 2^{[0,1]^n}$ to return $\{\vec 0\}$ everywhere and iteratively apply the Bellman operator, yielding arbitrarily precise approximations of $\achievable$ by $\bellman^k(\lb)$ as $k\to\infty$ \cite{DBLP:conf/mfcs/ChenFKSW13}\footnote{Precisely, $\lim_{k\to\infty}\bellman^k(\lb)\subseteq \achievable\subseteq \dwc(\pareto) = \cl(\lim_{k\to\infty}\bellman^k(\lb))$ where $\cl$ is the standard closure in $\Reals^n$.}.
Moreover, for every state $s$ it can be checked that the set $\bellman^k(\lb)(s)$ is presented at each step $k$ as a closed downward-closed convex polyhedron, i.e. a \emph{finite} object. Thus we can effectively construct any desired approximation.

However, it is not known how to bound the difference of the actual achievable set $\achievable$ and the approximation after $k$ iterations.
For that reason, \cite{cav18} introduced for the single-dimensional case the \emph{bounded value iteration} (named along the tradition of \cite{BRTDP}), a way to compute also an over-approximating sequence.
If we initialize $\ub$ to return $\dwc(\{\vec 1\})$ everywhere\footnote{The same holds even if we initialize to $0$ all the dimensions $i$ in states from which there is no path to $T_i$, as is customary in MDP analysis. The solution of \cite{cav18} is not sensitive to this and does not require this special treatment in the initialization of $\ub$.}, then $\lim_{k\to\infty}\bellman^k(\ub)$ is a fixpoint, which is generally different from the least one.
Hence \cite{cav18} modifies $\bellman$ so that it has a single fixpoint equal to the least one of the original $\bellman$.
Then both the sequence of lower bounds and of upper bounds converge to $\achievable$, the value of the game.
The modification is demonstrated in the next section, where we also illustrate the main ideas how to cope with the multi-dimensional case.

\section{Example}\label{sec:example}

In this section, we illustrate the issues preventing convergence of the upper bounds, as well as the solution of \cite{cav18} and our extension of it.
Value iteration converges if the SG is stopping, i.e. if the game reaches a designated sink with probability 1, or equivalently, if there are no end components (ECs). 
Hence the difficulty in solving reachability SG is rooted in ECs, as it is possible to cycle in its states infinitely long.
As a running example, consider the EC in Fig.~\ref{fig:standard-example} with states $\mathsf{p},\mathsf{q},\mathsf{r}$ and actions $\mathsf a,\ldots,\mathsf g$.
The symbols $\alpha$, $\beta$ and $\gamma$ are placeholders; in the single dimensional case, they represent a real number; in the multi dimensional case, a Pareto frontier. 
One can make this game a standard SG in the single dimensional case by, for example, replacing $\alpha$ with a transition that reaches the target with probability $\alpha$ and the sink with probability $1-\alpha$. The multi-dimensional case is a straightforward extension.

We start by considering the single-reachability objective. The standard Bellman update procedure as described in Section \ref{sec:bvi} reduces to the following equations, where intersections become minima and unions become maxima.
We write $\ub<k>$ as short for $\bellman^k(\ub)$.

\begin{align*}
\ub<i+1>(\mathsf{p}) &= \min{\{  \ub<i>(\mathsf{q}), \ub<i>(\mathsf{r}), \gamma \}} \\
\ub<i+1>(\mathsf{q}) &= \max{\{ \ub<i>(\mathsf{p}), \alpha \}} \\
\ub<i+1>(\mathsf{r}) &= \max{\{ \ub<i>(\mathsf{p}), \beta \}}
\end{align*}
By replacing $\ub$ with $\lb$, we get the update equations for the lower bound.
Recall that we initialize $\lb<0>$ to return $0$ everywhere and $\ub<0>$ to return $1$ everywhere.

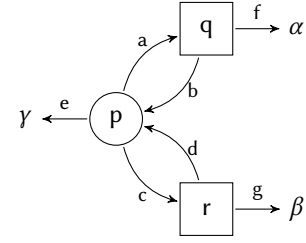
\begin{figure}
	\centering
	\begin{tikzpicture}[scale=0.6]
	
	\drawcirc (A) at (0,0) {$\mathsf{p}$};
	\drawbox (B) at (2,2) {$\mathsf{q}$};
	\drawbox (C) at (2,-2) {$\mathsf{r}$};
	
	\node (Bo) at (4,2) {$\alpha$};
	\node (Co) at (4,-2) {$\beta$};
	\node (Ao) at (-2, 0) {$\gamma$};
	
	\draw[->] (A) to[bend left]  node [midway,anchor=south, font=\footnotesize] {$\mathsf{a}$} (B);
	\draw[->] (B) to[bend left]  node [midway,anchor=north, pos=0.2, font=\footnotesize] {$\mathsf{b}$} (A);
	\draw[->] (A) to[bend right]  node [midway,anchor=south, below, font=\footnotesize] {$\mathsf{c}$} (C);
	\draw[->] (C) to[bend right]  node [midway,anchor=south, pos=0.2, above, font=\footnotesize] {$\mathsf{d}$} (A);
	
	\draw[->] (A) to  node [midway,anchor=south, font=\footnotesize] {$\mathsf{e}$} (Ao);
	\draw[->] (B) to  node [midway,anchor=south, font=\footnotesize] {$\mathsf{f}$} (Bo);
	\draw[->] (C) to  node [midway,anchor=south, font=\footnotesize] {$\mathsf{g}$} (Co);
	
	\end{tikzpicture}
	\caption{An example demonstrating the complications arising in an end component.}
	\label{fig:standard-example}
\end{figure}

\subsection{MDP}\label{sec:exMDP}
Firstly, let us briefly mention the solution of  \cite{atva,BVI} for MDP.
Suppose that all states belonged to the maximizing player, i.e. $\mathsf{p}$ was also maximizing.
Then, the initialization $\ub<0>=\vec 1$ is already a fixpoint, although the true value of all three states is $\max\{\alpha,\beta,\gamma\}$.
Intuitively, the reason for this is that the equations create a cyclic dependency: the process of finding the value by ``asking neighbours'' is not well-founded and all states falsely believe that they can achieve the higher value $1$.
\cite{cav18} calls such an EC \emph{bloated}, having an unjustifiably large (bloated) upper bound.
The solution of \cite{atva,BVI} is to detect that this is an EC and collapse it into a single state, eliminating the cycle. 
Only outgoing actions $\alpha,\beta,\gamma$ of the EC are kept, and in the next iteration, the Bellman operator correctly sets the value of the collapsed state to $\max\{\alpha,\beta,\gamma\}$, thus converging to the true value.
The solution of \cite{cav18} captures this idea from a different perspective: 
It does not change the underlying graph, but instead realizes that all three states can reach the \emph{``best exit''} of the EC, i.e.\ the state with an action exiting the EC and having the highest value.
Then the algorithm reduces the upper bounds of the states of the EC to that of the best exit. 
This is called \emph{deflating}, as the ``internal higher pressure'' of bloated upper bounds is ``relieved'', equalizing with the best exit.

\subsection{Single-reachability SG}\label{sec:exSingle}

Secondly, for single-reachability \emph{SG}, the EC cannot in general be collapsed, since the values of the states differ, and it is not clear a priori which states share a value.
They depend on the \emph{ordering} of the values of the exits, i.e. on the ordering of $\alpha$, $\beta$ and $\gamma$.

\textit{Case 1}: If $\gamma < \min(\alpha, \beta)$, then after the first iteration we have 
$\ub<1>(\mathsf{p}) = \gamma$, $\ub<1>(\mathsf{q}) = 1$ and $\ub<1>(\mathsf{r}) = 1$. 
After the next iteration, $\ub<2>(\mathsf{p}) = \gamma$, $\ub<2>(\mathsf{q}) = \alpha$ and $\ub<2>(\mathsf{r}) = \beta$. 
These are the true values, as observable in Figure \ref{fig:standard-example}.
In this case $\ub<k>$ converges to the value.
However, note that the values of the states in the same EC are different. 

\textit{Case 2}: If $\gamma \geq \min(\alpha, \beta)$, and say $\alpha > \beta$, then the values of $\mathsf{p}$ and $\mathsf{r}$ are $\beta$ and that of 
$\mathsf{q}$ is $\alpha$. This is the case, because $\mathsf{p}$ will always play action $\mathsf{c}$, not allowing state $\mathsf{r}$ to achieve anything but the smallest value $\beta$. 
However, $\ub<k>$ does not converge to these values. 
In the first iteration, $\ub<1>(\mathsf{p}) = \gamma$, $\ub<1>(\mathsf{q}) = 1$ and $\ub<1>(\mathsf{r}) = 1$. 
After the next iteration, $\ub<2>(\mathsf{p}) = \ub<2>(\mathsf{q}) = \ub<2>(\mathsf{r}) = \gamma$.
After this, the upper bounds do not change any more, because we have the problem of cyclic dependencies as described in Section \ref{sec:exMDP}.
If we fix the strategy of the Minimizer to $\mathsf{c}$ as that is the best choice, only $\{\mathsf{p},\mathsf{r}\}$ forms an EC.
The value of both $\mathsf{p}$ and $\mathsf{r}$ is $\beta$, as that is the best exit that the Maximizer can achieve, given that Minimizer does not play the suboptimal action $\mathsf{a}$.
Such an EC where all states share the same value is called simple end component (SEC) \cite{cav18}. 
It is simple, because after fixing the strategy of Minimizer to be optimal, this player cannot influence the play anymore (as the SG locally becomes an MDP). In the SEC, Maximizer can direct the play to the best exit and almost surely achieve the value of it.
Deflating the SEC $\{\mathsf{p}, \mathsf{r}\}$, i.e. setting the upper bound for all states in the SEC to that of the best exit, correctly updates the bounds to $\beta$. 
Afterwards, the upper bound of $\mathsf{q}$ is correctly set to $\alpha$ in the next iteration.
So one would like to find and deflate all SECs. 

However, which states form a SEC depends on the relative ordering of the exits' values and the corresponding choices that Minimizer makes (recall we had to fix the strategy of $\mathsf{p}$ to the optimal action $c$ in order to realize which states form the SEC). 
Indeed, in the case with $\alpha<\beta$, a different SEC ($\{\mathsf{p},\mathsf{q}\}$) should be deflated and if $\alpha=\beta$ then all three states form a SEC.
Since we do not know the values of the exits, the algorithm uses the approximations ($\lb<i>$) to guess which actions are suboptimal for the Minimizer, and hence which states form a SEC. 
As the lower approximation converges to the value, the true SECs are eventually detected and correctly deflated.
However, when $\lb<i>$ is not yet close enough to the value, the computation of SECs can be wrong, e.g. if $\alpha < \beta$, but for the first few iterations of the algorithm the lower bound on $\beta$ is smaller than that on $\alpha$. 
Then, for these first iterations, the algorithm believes $\{\mathsf{p},\mathsf{r}\}$ to be the SEC, and only afterwards realizes that it actually is $\{\mathsf{p},\mathsf{q}\}$.
Hence, the operation we perform on the SEC has to be conservative, i.e. sound even if it is given a set of states that actually do not form a SEC.
This is why deflating was introduced, as it is sound for any EC, even ones that are not SECs \cite[Lemma 3]{cav18}. In contrast, modifying the underlying graph by collapsing as in \cite{atva,BVI} would commit to the detected SEC-candidate and thereby possibly make the wrong choice.
Note that we never know that we have correctly detected a SEC, we just know that in the limit we will eventually detect it.

\subsection{Generalized-reachability SG}\label{sec:exMOSG}
\begin{figure}
	\centering
	
	\begin{tikzpicture}[scale=1.4, font=\footnotesize]
	\draw[draw=none,-] (-0.4, -0.4) -- (-0.4, 1.2) -- (1.2, 1.2) -- (1.2, -0.4) -- (-0.4, -0.4);
	
    \filldraw[fill=\myblue, draw=none] (0,0) -- (0, 0.9) -- (0.5, 0.9) -- (0.5, 0);
    
	\draw[->]  (0, -0.1) -- (0, 1.15);
	\draw[->]  (-0.1, 0) -- (1.15, 0);
	
	\draw[-]   (0, 0.9) -- (0.5, 0.9) -- (0.5, 0);
	
	\draw[fill=black,radius=0.75pt] (0, 0.9) circle node [left] {0.9};
	\draw[fill=black,radius=0.75pt] (0.5, 0) circle node [below] {0.5};
	
	\draw[fill=black,radius=0.75pt] (0.5, 0.9) circle node [right] {(0.5, 0.9)};
	\end{tikzpicture}
	\hfill
	\begin{tikzpicture}[scale=1.4, font=\footnotesize]
	\draw[draw=none,-] (-0.4, -0.4) -- (-0.4, 1.2) -- (1.2, 1.2) -- (1.2, -0.4) -- (-0.4, -0.4);
	
    \filldraw[fill=\myred, draw=none] (0,0) -- (0, 0.5) -- (0.9, 0.5) -- (0.9, 0);
    
	\draw[->]  (0, -0.1) -- (0, 1.15);
	\draw[->]  (-0.1, 0) -- (1.15, 0);
	
	\draw[-]   (0, 0.5) -- (0.9, 0.5) -- (0.9, 0);
	
	\draw[fill=black,radius=0.75pt] (0, 0.5) circle node [left] {0.5};
	\draw[fill=black,radius=0.75pt] (0.9, 0) circle node [below] {0.9};
	
	\draw[fill=black,radius=0.75pt] (0.9, 0.5) circle node [above] {(0.9, 0.5)};
	\end{tikzpicture}
	\hfill
	\begin{tikzpicture}[scale=1.4, font=\footnotesize]
	\draw[draw=none,-] (-0.4, -0.4) -- (-0.4, 1.2) -- (1.2, 1.2) -- (1.2, -0.4) -- (-0.4, -0.4);
	
	\draw[->]  (0, -0.1) -- (0, 1.2);
	\draw[->]  (-0.1, 0) -- (1.2, 0);
	
	\draw[-]   (0, 1) -- (0.5, 0.9) -- (0.9, 0.5) -- (1, 0);
	
	\draw[fill=black,radius=0.75pt] (0.5, 0.9) circle node [right] {(0.5, 0.9)};
	\draw[fill=black,radius=0.75pt] (0.9, 0.5) circle node [right] {(0.9, 0.5)};
	\draw[fill=black,radius=0.75pt] (0, 1) circle node [left] {1};
	\draw[fill=black,radius=0.75pt] (1, 0) circle node [below] {1};
	\end{tikzpicture}
	\caption{Pareto frontiers of $\alpha$ (left), $\beta$ (center) and $\gamma$ (right) in a 2-objective setting. X-axis represents the value along the first objective and Y-axis represents the value along the second objective.} \label{fig:Pareto-objective-example}
\end{figure}
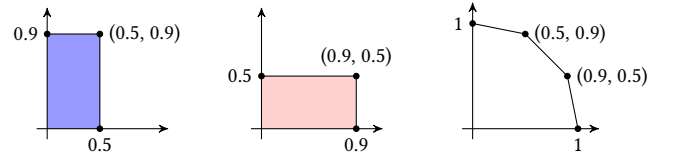

\colorlet{mix}{red!50!blue}

\begin{figure}[t]
	\centering
	%
	%
	%
	%
	%
	%
	\begin{subfigure}{0.25\textwidth}
		\centering
		\begin{tikzpicture}[scale=1.4, font=\footnotesize]
		\draw[draw=none,-] (-0.4, -0.4) -- (-0.4, 1.9) -- (2.5, 2.5) -- (1.9, -0.4) -- (-0.4, -0.4);
		
		\filldraw[fill=\myblue, draw=none] (0,0) -- (0.5, 0.5) -- (0.5, 0) -- (0, 0);
		\filldraw[fill=\myred, draw=none] (0,0) -- (0.5, 0.5) -- (0, 0.5) -- (0, 0);
		
		\draw[-]  (0, -0.1) -- (0, 1.15);
		\draw[-]  (-0.1, 0) -- (1.15, 0);
		
		\draw[-] (2,0)--(0,2);
		\draw[-,\myblue,ultra thick] (2,0)--(1,1);
		\draw[-,\myred,ultra thick] (0,2)--(1,1);
		\draw[fill=black!50,radius=2pt] (1, 1) circle node [left] {};
		
		\draw[-]   (0, 0.9) -- (0.5, 0.9) -- (0.5, 0);
		\draw[-]   (0, 0.5) -- (0.9, 0.5) -- (0.9, 0);
		
		\draw[dotted,-stealth]   (0, 0) -- (0, 2.2) node [above] {$\dir_1$};
		\draw[dotted,->]   (0, 0) -- (1.2, 1.2) node [above right] {$\dir_2$};
		\draw[dotted,->]   (0, 0) -- (2.2, 0) node [right] {$\dir_3$};
		
		\draw[thick, -, black!50] (0, 0) -- (0.5, 0.5);
		
		\draw[fill=black,radius=1.pt] (0, 0.9) circle node [left] {0.9};
		\draw[fill=black,radius=1.pt] (0.5, 0) circle node [below] {0.5};
		\draw[fill=black,radius=1.pt] (0, 0.5) circle node [left] {0.5};
		\draw[fill=black,radius=1.pt] (0.9, 0) circle node [below] {0.9};
		\draw[fill=black,radius=1.pt] (0.5, 0.5) circle;
		\draw[fill=black,radius=1.pt] (0.9, 0.5) circle;
		\draw[fill=black,radius=1.pt] (0.5, 0.9) circle;
		\end{tikzpicture}
		\caption{} \label{fig:vis-regions}
	\end{subfigure}
	\begin{subfigure}{0.15\textwidth}
		\centering
		\begin{tikzpicture}[scale=1.4, font=\footnotesize]
		\draw[draw=none,-] (-0.4, -0.4) -- (-0.4, 1.9) -- (2.5,2.5) -- (1.9, -0.4) -- (-0.4, -0.4);
		
		\draw[->]  (0, -0.1) -- (0, 1.15);
		\draw[->]  (-0.1, 0) -- (1.15, 0);
		
		\draw[-]   (0, 0.5) -- (0.5, 0.5) -- (0.5, 0);
		
		\draw[fill=black,radius=1pt] (0.5, 0) circle node [below] {0.5};
		\draw[fill=black,radius=1pt] (0, 0.5) circle node [left] {0.5};
		\draw[fill=black,radius=1pt] (0.5, 0.5) circle;
		\end{tikzpicture}
		\caption{} \label{fig:deflate-regions}
	\end{subfigure}
	\caption{(a) Visualizing the regions for state $\mathsf{p}$; and (b) the result of deflating. } \label{fig:regions}
	
\end{figure}
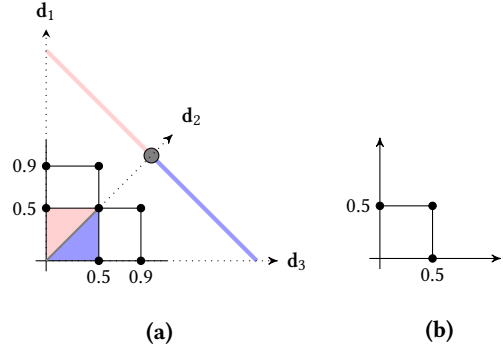

Here we intuitively describe and illustrate the main elements of our solution.
The formal definitions of the key concepts only follow in the next section.

\paragraph{Regions.}
Consider again the example of Fig.~\ref{fig:standard-example}.
In the multi-dimensional case, instead of $\alpha, \beta$ and $\gamma$ being reals, they are sets of achievable vectors.
Let them be given as in Fig.~\ref{fig:Pareto-objective-example}, so e.g. $\alpha = \dwc(\{(0.5,0.9)\})$.
Here $\gamma$ gives the highest values, so it is the best one for Maximizer, and hence Minimizer will not play the corresponding $\mathsf{e}$ (as in Case 2 in Section \ref{sec:exSingle}).
$\alpha$ and $\beta$, however, cannot be compared. Depending on the trade-off (corresponding to a direction) that Maximizer wants to achieve, $\alpha$ or $\beta$ might be better than the other.
To this end, let $\dir$ be the direction in which Maximizer wants to maximize.
Depending on $\dir$, Minimizer's behaviour changes. 
If the objective along the x-axis is more important, then Minimizer chooses action $\mathsf a$. 
This way, the value of the more important objective is restricted to 0.5.
If on the other hand, the objective along y-axis is more important, then the Minimizer chooses action $\mathsf c$. 
The Minimizer, for each direction $\dir$, decides on the action to be chosen by comparing $\alpha$ and $\beta$ evaluated in that direction; in other words, by computing the minimum of $\alpha[\dir]$ and $\beta[\dir]$.

Our algorithm identifies finitely many \textit{regions} where the Minimizer has the same preference ordering over actions and then we deflate each region separately. 
In our example, we can identify three regions, as shown in Fig.~\ref{fig:vis-regions}. 
Between the directions ${\dir_1}$ and ${\dir_2}$ (red \tikz\draw[draw=none,fill=\myred] (0,0) circle (.5ex); region),
Minimizer's best choice is action $c$; 
between ${\dir_2}$ and ${\dir_3}$ (blue \tikz\draw[draw=none,fill=\myblue] (0,0) circle (.5ex); region),
Minimizer's best choice is action $a$; 
and along ${\dir_2}$ (grey \tikz\draw[draw=none,fill=black!50] (0,0) circle (.5ex); line), Minimizer is indifferent.

\paragraph{Deflating regional SECs.}
Once restricting to a region fixes the preference ordering over Minimizer's actions, we can proceed as in the single-dimensional case:
We fix Minimizer's optimal strategy based on the lower bounds, identify SEC-candidates and deflate them.
That means we update the Pareto frontier in the region to that of the \emph{best exit} from the SEC. The whole Pareto frontier is constructed piece by piece, region by region.


Returning to our running example, we have already identified the three regions in the Pareto frontier for state $\mathsf{p}$ in Figure \ref{fig:vis-regions}.
The SECs depending on the regions are as follows: 
In the blue \tikz\draw[draw=none,fill=\myblue] (0,0) circle (.5ex); region it is $\{\mathsf{p},\mathsf{q}\}$, in the red \tikz\draw[draw=none,fill=\myred] (0,0) circle (.5ex); region it is $\{\mathsf{p},\mathsf{r}\}$, and along $\dir_2$ all three states form a SEC. 
Deflating the blue \tikz\draw[draw=none,fill=\myblue] (0,0) circle (.5ex); region, we see that the best exit from the SEC has value $\alpha$, so between 0\textdegree\ and 45\textdegree\ the value of $\mathsf{p}$ is set to the corresponding part of $\alpha$.
Doing the same for the other two regions results in the Pareto frontier depicted in Figure \ref{fig:deflate-regions}.
This result is also intuitively expected, as depending on which direction Maximizer prefers, Minimizer can always restrict the play to the other exit.
Note that for the sake of example here we always talked about the true values, while the algorithm does not know these precisely. 
Therefore, deflating cannot update the values based on the value of $\alpha$, but only on its approximation.
Being on the safe side, the values will be decreased only to its over-approximation.

\medskip

\begin{figure}
	\centering
	\begin{subfigure}{0.22\textwidth}
		\centering
		\includegraphics[width=\linewidth]{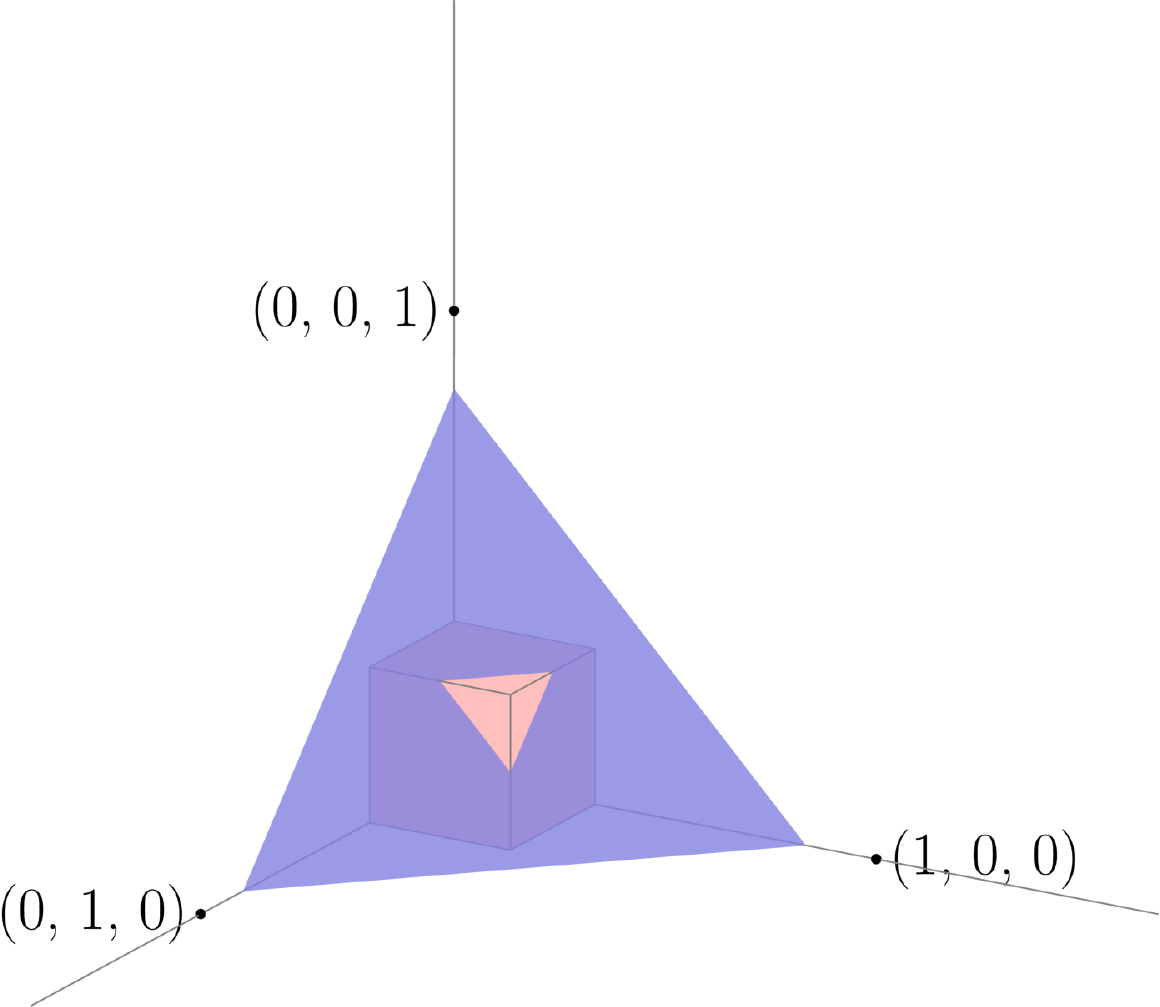}
	\end{subfigure}%
	\begin{subfigure}{.17\textwidth}
		\centering
		\begin{tikzpicture}[font=\footnotesize, scale=0.8]
		\draw[draw=none] (-2, -0.33) rectangle (2, 3.5);
		\coordinate (A) at (-1.5, 0);
		\coordinate (B) at (1.5, 0);
		\coordinate (C) at (0, 2.598);
		\coordinate (D) at (-0.35, 1.1);
		\coordinate (E) at (0.35, 1.1);
		\coordinate (F) at (0, 0.566);
		\node[below] at (A) {[0, 1, 0]};
		\node[below] at (B) {[1, 0, 0]};
		\node[above] at (C) {[0, 0, 1]};
		\draw[gray!50, -] (A) -- (B) -- (C) -- (A);
		\draw[thick,-] (D) -- (E) -- (F) -- (D);
		\end{tikzpicture}
	\end{subfigure}%
	
	\centering
	\begin{subfigure}{0.22\textwidth}
		\centering
		\includegraphics[width=\linewidth]{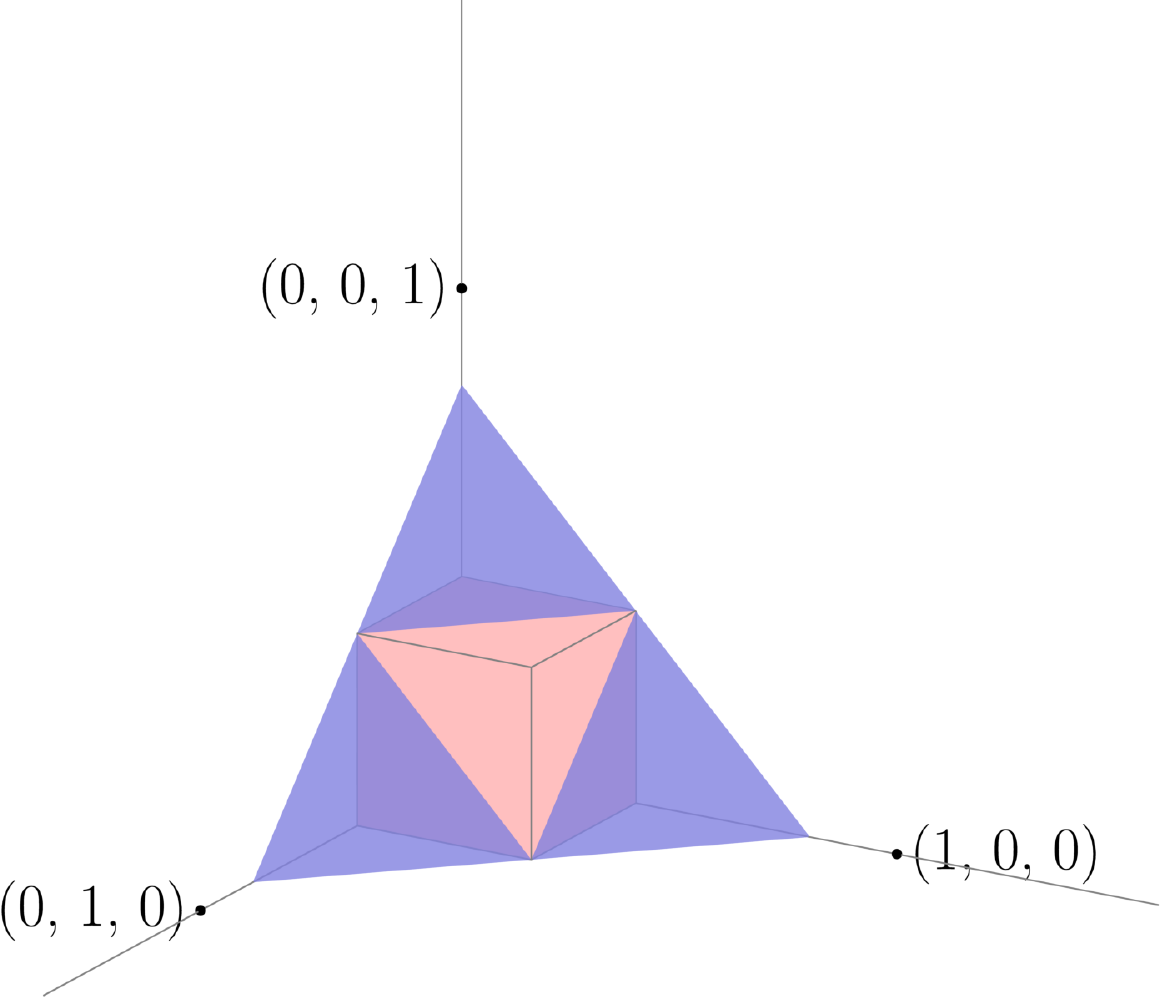}
	\end{subfigure}%
	\begin{subfigure}{.17\textwidth}
		\centering
		\begin{tikzpicture}[font=\footnotesize,scale=0.8]
		\draw[draw=none] (-2, -0.33) rectangle (2, 3.5);
		\coordinate (A) at (-1.5, 0);
		\coordinate (B) at (1.5, 0);
		\coordinate (C) at (0, 2.598);
		\coordinate (D) at (0, 0);
		\coordinate (E) at (0.75, 1.3);
		\coordinate (F) at (-0.75, 1.3);
		\node[below] at (A) {[0, 1, 0]};
		\node[below] at (B) {[1, 0, 0]};
		\node[above] at (C) {[0, 0, 1]};		
		\draw[gray!50, -] (A) -- (B) -- (C) -- (A);
		\draw[thick,-] (D) -- (E) -- (F) -- (D);
		\end{tikzpicture}
	\end{subfigure}%
	
	\centering
	\begin{subfigure}{0.22\textwidth}
		\centering
		\includegraphics[width=\linewidth]{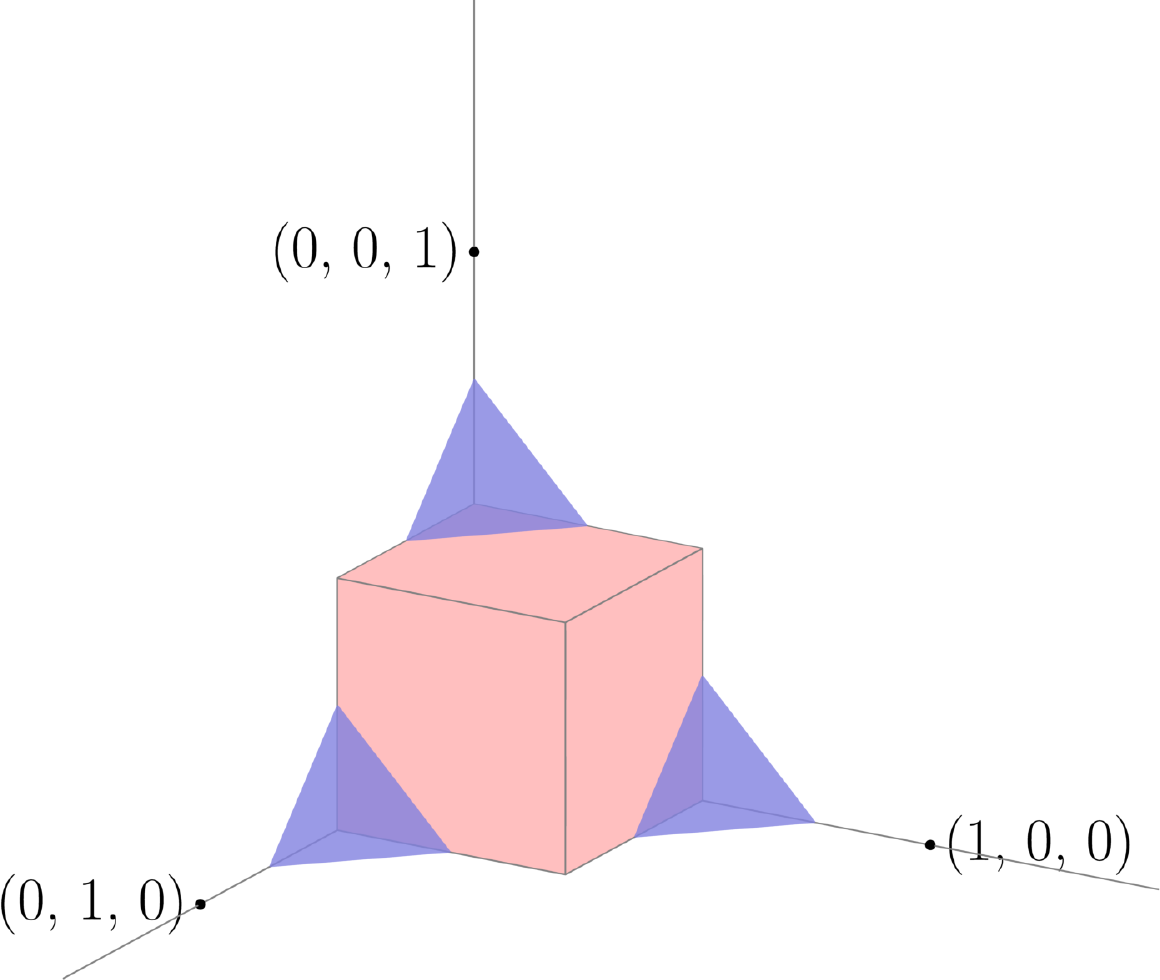}
	\end{subfigure}%
	\begin{subfigure}{.17\textwidth}
		\centering
		\begin{tikzpicture}[font=\footnotesize,scale=0.8]
		\draw[draw=none] (-2, -0.33) rectangle (2, 3.5);
		\coordinate (A) at (-1.5, 0);
		\coordinate (B) at (1.5, 0);
		\coordinate (C) at (0, 2.598);
		\coordinate (D) at (-0.6,0);
		\coordinate (E) at (0.6, 0);
		\coordinate (F) at (1.05, 0.7785);
		\coordinate (G) at (0.45, 1.8186);
		\coordinate (H) at (-0.45, 1.8186);
		\coordinate (I) at (-1.05, 0.7785);
		\node[below] at (A) {[0, 1, 0]};
		\node[below] at (B) {[1, 0, 0]};
		\node[above] at (C) {[0, 0, 1]};
		\draw[gray!50, -] (A) -- (B) -- (C) -- (A);
		\draw[thick,-] (D)  (E) -- (F)  (G) -- (H)  (I) -- (D);
		\end{tikzpicture}
	\end{subfigure}%
	\caption{The left column shows Pareto frontiers; there is the blue tetrahedon and a pink box of varying size. The right column shows the projection of the intersection onto the projective hyperplane.}
	\label{fig:type2}
\end{figure}
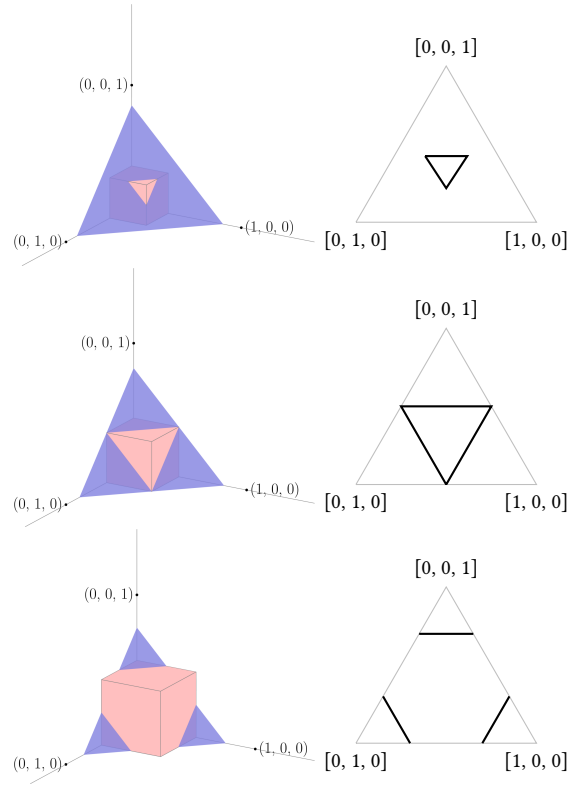
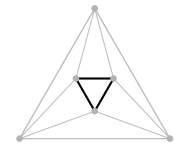
\begin{figure}
	\centering
	\begin{tikzpicture}[scale=1]
	\draw[draw=none] (-2, -0.33) rectangle (2, 1.8);
	\coordinate (A) at (-1, 0);
	\coordinate (B) at (1, 0);
	\coordinate (C) at (0, 1.732);
	\coordinate (D) at (-0.25, 0.8);
	\coordinate (E) at (0.25, 0.8);
	\coordinate (F) at (0, 0.366);
	
	\draw[gray!50, -] (A) -- (B) -- (C) -- (D) -- (E) -- (F);
	\draw[gray!50, -] (A) -- (F) -- (D) (A) -- (D) (A) -- (C);
	\draw[gray!50, -] (E) -- (C) (E) -- (B) (F) -- (B);
	
	\draw[gray!50, -] (A) -- (B) -- (C) -- (A);
	\draw[thick,-] (D) -- (E) -- (F) -- (D);
	
	\draw[draw=none,fill=black!30,radius=1.25pt] (A) circle;
	\draw[draw=none,fill=black!30,radius=1.25pt] (B) circle;
	\draw[draw=none,fill=black!30,radius=1.25pt] (C) circle;
	\draw[draw=none,fill=black!30,radius=1.25pt] (D) circle;
	\draw[draw=none,fill=black!30,radius=1.25pt] (E) circle;
	\draw[draw=none,fill=black!30,radius=1.25pt] (F) circle;
	\end{tikzpicture}
	\caption{Triangulation of the top right of Fig.~\ref{fig:type2}}
	\label{fig:triangulation-example}
\end{figure}


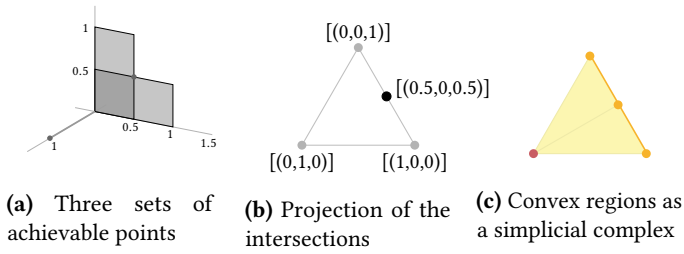
\begin{figure}[t]
	\hspace*{-0.5cm}
	\begin{subfigure}{0.15\textwidth}
		\begin{tikzpicture}[scale=0.42, font=\Large]
		\begin{axis}[
		axis equal,
		axis lines=middle,
		xmin=0, xmax=1.25, 
		ymin=0, ymax=1.25, 
		zmin=0, zmax=1.25, 
		view={120}{20},
		axis line style={gray!60, -},
		xtick style={draw=none},
		ytick style={draw=none},
		ztick style={draw=none},
		area plot/.style={
			-,
			fill opacity=0.75,
			fill=black!30!white,
			mark=none,
		},
		area plot dark/.style={
			-,
			fill opacity=0.75,
			fill=black!40!white,
			mark=none,
		},
		mymark/.style = {mark=*, color=black!60, -},
		myline/.style = {-, ultra thick, black!40!white}
		]
		\addplot3 [area plot] table [x expr=0, y=a, z=b] {
			a b
			0 0
			0 1
			0.5 1
			0.5 0.5
			1 0.5
			1 0
			0 0
		};
		\addplot3 [area plot dark] table [x expr=0, y=a, z=b] {
			a b
			0 0
			0 0.5
			0.5 0.5
			0.5 0
			0 0
		};
		\addplot3 [mymark] table [x expr=0, y=a, z=b] {
			a b
			0.5 0.5
		};
		\addplot3 [mymark] table [x =a, y expr =0, z expr=0] {
			a 
			1
		};
		\addplot3 [myline] table [x=a, y=b, z expr=0] {
			a b
			1 0
			0 0
		};
		\end{axis}
		\draw[draw=none] (0, 0) rectangle (7, 7);
		\end{tikzpicture}
		\caption{Three sets of achievable points}
		\label{fig:proj-sets}
	\end{subfigure}
	~~~
	\begin{subfigure}{.15\textwidth}
		\begin{tikzpicture}[scale=0.75,font=\footnotesize]
		\draw[draw=none] (-2, -0.33) rectangle (2, 3.5);
		\draw[-, gray!50] (-1,0) -- (1, 0) -- (0, 1.732) -- (-1, 0);
		\draw[fill=black,radius=2.25pt] (0.5, 0.866) circle;
   		\draw[draw=none,fill=black!30,radius=2.25pt] (-1, 0) circle;
        \draw[draw=none,fill=black!30,radius=2.25pt] (1, 0) circle;
        \draw[draw=none,fill=black!30,radius=2.25pt] (0, 1.732) circle;
		\node at(1,-0.3) {[(1,0,0)]};
		\node at(-1,-0.3) {[(0,1,0)]};
		\node at(0,2) {[(0,0,1)]};
		\node at(1.5,1) {[(0.5,0,0.5)]};
		
		\end{tikzpicture}
		\caption{Projection of the intersections}
		\label{fig:proj-plane}
	\end{subfigure}%
	~~~
	\begin{subfigure}{.15\textwidth}
		\begin{tikzpicture}[scale=0.75]
		\draw[draw=none] (-2, -0.33) rectangle (2, 3.5);
        
        \draw[-,very thick,saffron] (0, 1.732) -- (1, 0);
        
        \draw[-,fill=bananamania,draw=none] (-1,0) -- (1, 0) -- (0, 1.732) -- (-1, 0);
        
        \draw[-,draw=bananamania!95!black] (0, 1.732) -- (-1,0) -- (1,0);
		\draw[-, bananamania!95!black] (0.5, 0.866) -- (-1, 0);
		
		
		\draw[fill=saffron,draw=none,radius=2.25pt] (0.5, 0.866) circle;
		\draw[draw=none,fill=indianred,radius=2.25pt] (-1, 0) circle;
		\draw[draw=none,fill=saffron,radius=2.25pt] (1, 0) circle;
		\draw[draw=none,fill=saffron,radius=2.25pt] (0, 1.732) circle;
       
		\end{tikzpicture}
		\caption{Convex regions as a simplicial complex}
		\label{fig:proj-partit}
	\end{subfigure}
	\caption{Projections of intersections of Pareto frontiers to the projection plane, which in 3D is the triangle formed by the points $(1,0,0)$, $(0,1,0)$ and $(0,0,1)$. In Fig. \ref{fig:proj-plane}, labels represent directions and not individual vectors.}
	\label{fig:proj}
\end{figure}

\paragraph{Computing and representing regions.}
As explained above, a region depends on the preference ordering of actions.
To compute regions where this ordering is constant, we use geometric methods.
In the example of Fig.~\ref{fig:vis-regions}, the point where the preference ordering changes is $(0.5,0.5)$, which is where the two Pareto frontiers intersect.
So, intuitively, by drawing the Pareto frontiers and finding the points of intersection, we can identify the regions (sets of the corresponding directions) where the preference ordering over actions is constant.

In Figure~\ref{fig:type2}, we give a set of three examples to illustrate the construction of regions.
The left picture in each row of the figure shows two Pareto frontiers: One is the blue tetrahedron, generated by Maximizer's free, but exclusive choice between target sets.
The other is a red box of different sizes, generated by the possibility to reach a state in all target sets with a given probability. 
From top to bottom, we increase this probability, thereby increasing the size of the box, yielding three different examples.
We define regions as sets of directions.
In order to draw directions, it is useful to consider the so-called \emph{projective hyperplane}.
It is the set of all directions and can be drawn (in our case with non-negative vectors only) as a triangle with corners $[1,0,0],[0,1,0],[0,0,1]$, capturing all directions.
When a point (vector) $\vec v$ is projected into its direction $[\vec v]$, it intuitively corresponds to drawing a ray from the origin through the point $\vec v$.
If we identify the projective hyperplane with the hyperplane passing through the \emph{points} $(1,0,0)$, $(0,1,0)$ and $(0,0,1)$ (or more precisely with this triangle) then the intersection of the ray and the hyperplane, say point $\vec p_v$, is the \emph{projection} of $\vec v$ to the projective hyperplane. 
In our example, the right side of the figure shows the projection of the intersection of the Pareto frontiers onto the projective hyperplane.
This gives rise to three regions, each with different preference ordering: the inner open triangle, its boundary and the outer triangle with the hole.
Minimizer prefers the red action in the outside triangle, the blue one in the inside triangle, and is indifferent on the boundary.
As these regions are hard to describe (as well as possibly not even convex and connected), we  triangulate the projections to get smaller regions which are convex and generated by finitely many points. 
The triangulation of the top right of Figure \ref{fig:type2} is depicted in Figure~\ref{fig:triangulation-example}.
Further note that while the preference ordering of actions is constant in each region, the faces of a region represent turning points of the preference ordering; hence these faces need to be separate regions like is customary for timed automata~\cite{DBLP:journals/tcs/AlurD94}. 
Hence in order to represent the regions, we thus decompose the triangle (generally, in higher dimensions, a \emph{simplex}) into open triangles, open line segments and points (in general into a \emph{simplicial complex}, i.e. the simplex together with its faces and recursively their faces).

As another example of the projection to the projective hyperplane and the triangulation, consider Figure~\ref{fig:proj-sets} with three achievable sets: two rectangles -- $\dwc(\{(1, 0, 0.5)\})$ and $\dwc(\{(0.5, 0, 1)\})$; and one line -- $\dwc(\{(0, 1, 0)\})$.
The frontiers of the sets generate only one non-empty intersection\footnote{The neutral element $\{\vec 0\}$ is not considered a non-empty intersection.}, namely the point $(0.5,0,0.5)$.
Its projection is represented by its direction, $[(0.5, 0, 0.5)]$ in Figure \ref{fig:proj-plane}.
In order to keep the representation of regions effective, we again triangulate regions into finer ones, which are convex and generated by finitely many points, see Fig.~\ref{fig:proj-partit}. 
Finally, note that Pareto frontiers of smaller dimensions may induce regions that are faces of the projective hyperplane (triangle). In this example, the vertex at $[(0,1,0)]$ is its own region, as it is the only direction where playing the line-action is not optimal for Minimizer.
We can also see that in Fig.~\ref{fig:deflate-regions}: the red vertex corresponds to Minimizer choosing one of the "rectangular" actions (as the other action is suboptimal), the orange region to choosing the action yielding $(0,1,0)$, and in the yellow Minimizer is indifferent between all actions.
Since these cases only arise on faces of the projective hyperplane, the decomposition into the simplicial complex of the projective hyperplane (triangle) caters for these corner cases. 
Note that for identifying the regions, we considered the point $[(0.5,0,0.5)]$, which is the turning point of preference between the two rectangles. 
As both of them are suboptimal in this direction, this is not necessary to get the coarsest partition. 
However, it is not a problem to use a finer partition (splitting the orange line and the yellow triangle), as we still have the invariant that in every region the strategy of Minimizer is constant.

\section{Algorithm}\label{sec:algo}

\subsection{Lifting the concepts from the single-dimensional case}\label{sec:defs}

Before giving the algorithm, we have to define extensions of the concepts of best exit and simple end component (SEC) introduced in \cite{cav18} to the multi-objective setting, as intuitively discussed in the previous section. 
To this end, we also introduce the concept of \emph{regions}.

\paragraph{Best exits}
In the single-dimensional case, the best exit of an EC was just the best exiting action for the Maximizer. 
In the multi-dimensional setting, Maximizer cannot only pick the best exit, but first visits all targets inside the (S)EC and then use any combination of exits to achieve any desired tradeoff. 
The definition of best exit depends on a parameter $f$. This function is used to calculate the set of achievable points from an exit. 
We can instantiate it with $\val$ to denote the actual set of achievable points, as well as with the over-approximation $\ub$; in the algorithm, we do the latter, as we do not know $\val$.

Thus we define the best exit in the multi-dimensional setting (similarly as $X(\state,\action)$ in Section \ref{sec:bvi}):
\[\exit[f](T) \eqdef \left(\dwc(\sum_{\state \in T} \{\mathbbm 1_{\targetsets}(\state)\}) + \conv(\ \bigcup_{\mathclap{(\state,\action) \in \exits(T_\Box)}} f(\state,\action))\right)\cap\mathbf 1\]
The first part ensures that, if a target is in the EC, all states in the EC have probability 1 to reach it; the second part takes the convex hull of the union of (Pareto sets of) all of Maximizer's exits, corresponding to randomizing over the exiting actions.
For general ECs, this may give a strict over-approximation since Minimizer might prevent Maximizer from freely visiting all states and combining all actions.
However, for SECs the expression is later shown exact.
Note that here we use the convention $\bigcup_\emptyset (\cdot) = \{\vec 0\}$, which is a neutral and minimal element. 
This solves the corner case of an EC without any exit.

\paragraph{Regions}
The extension of SEC works only when partitioning the set of all possible directions into \emph{regions}, and then applying the same ideas as in the single-dimensional case in each region separately.

\begin{definition}[Region]
	A \emph{region} is a subset $R \subseteq \mathcal \directions$ of directions. 
\end{definition}
To keep the presentation simple, we rely on a very general definition of regions at this point. We will see later in Section \ref{sec:getReg} how we can restrict to handling only regions
that correspond to a finitely generated cone.
In the following, slightly abusing notation, we sometimes view a region $R$ as the set of points it contains, i.e. $\{v \in [0, 1]^n \mid \exists \dir \in R: [v] = \dir\}$.

\paragraph{Simple ECs}
In the single-dimensional case, the idea of SEC is the following:
If Minimizer fixes their strategy to the optimal strategy (i.e. ignores all suboptimal actions), and in the remaining game there still exists an EC, then this EC is simple.
It is the best choice of Minimizer to allow Maximizer to roam around freely in the SEC and pick the best exit.
Thus, all states in the SEC have the same value, namely that of the best exit (recall, best for Maximizer).

In the multi-dimensional case, the optimal strategy of Minimizer depends on the tradeoffs between the different goals. 
This is why, to generalize the concept of SEC, we need to add the restriction that a set of states is a SEC for some region $R$, as the trade-offs between the goals are resolved in the same way in the whole region, or in other words: where the optimal strategy of Minimizer is the same for all directions in $R$.
Formally:
\begin{definition}[Regional $\SEC$]
	An EC $T$ is \emph{a regional simple end component for some \region\ $R$}, if for every direction $\dir \in R$ and all states $s \in T$, $\val(s)[\dir] = \exit[\val](T)[\dir]$.
\end{definition}
Note that from this definition we also know that all states in the regional $\SEC$ have the same value. 
Moreover, as we shall see, the definition implies that on $R$, the optimal strategy of Minimizer should be the same in all directions.
Lifting this to a set of regions we have the following property:

\begin{definition}[Consistent Partition]
	Let $T$ be an EC and $f \colon S \rightarrow 2^{\mathbb{R}^n}$. 
	A partition of the set $\directions$ of directions into a set of regions $\regions$ is called \emph{\consistent} w.r.t. $T$ and $f$ if for all $R \in \regions$ and all $\dir_1,\dir_2 \in R$, $s \in T_\circ$ and $a\in \Av(s)$ it holds that
	\begin{align*}
	f(s,a)[\dir_1] &= \min_{b \in \Av(s)} f(s,b)[\dir_1] \iff \\
	f(s,a)[\dir_2] &= \min_{b \in \Av(s)} f(s,b)[\dir_2].
	\end{align*}
\end{definition}

In the other direction, we shall see that every possible regional SEC can be defined on regions of an arbitrary consistent partition. 
Hence, algorithmically, we shall be looking for such partitions first and then for regional SECs.


\subsection{Algorithms}

We present our overall bounded VI procedure as Algorithm~\ref{alg:MOBVI}. In the following, we provide intuitive explanations of the algorithm and its sub-procedures, as well as the proofs for the lemmata on correctness of the sub-procedures. 
The correctness of the whole algorithm is proven in Section~\ref{sec:proofStruct}.
Section~\ref{sec:getReg} gives more details on the effectiveness of the computation in Algorithm~\ref{alg:getregions}, as that pseudocode is rather mathematical and it is not trivial to see that it is indeed effectively computable and yields an effective approximation.
 
\begin{algorithm}[t]
	\caption{Multi-Objective Bounded Value Iteration}\label{alg:MOBVI}
	\begin{algorithmic}[1]
		\Require 
			\Statex SG $\G$, generalized-reach. objective $\targetsets$, precision $\varepsilon$
		\Ensure 
			\Statex $\mathcal L, \mathcal U \text{ such that } \forall \dir\in\directions:\mathcal L[\dir]\leq \pareto[\dir]\leq\mathcal  U[\dir]$ and $\mathcal U[\dir]-\mathcal L[\dir]<\varepsilon$
		\Procedure{\textsf{MO-BVI}}{$\G$,$\targetsets$,$\varepsilon$}
		\For {each $\state \in \states$} \Comment{Initialization}
		\State $\lb(\state) \gets \set{\vec{0}}$ \Comment{to the least and}
		\State $\ub(\state) \gets \dwc(\set{\vec{1}})$ \Comment{the greatest values}
		\EndFor
		\medskip
		
		\Repeat \Comment{The new Bellman update $\mop$}
		\State $\lb\gets \bellman(\lb)$ \Comment{Standard Bellman updates}
		\State $\ub\gets \bellman(\ub)$ 
		\State $\ub \gets \doSECs(\G,\lb,\ub)$ \label{line:doSECs} \Comment{New treatment}
		\Until $\displaystyle\max_{\dir \in \mathcal{D}} \ub(\initstate)[\dir] - \lb(\initstate)[\dir] < \varepsilon$ \Comment{$\varepsilon$-approximate} \label{line:StoppingCrit}
		\State\Return $(\lb(\initstate),\ub(\initstate))$
		\EndProcedure
	\end{algorithmic}
\end{algorithm}

\begin{algorithm}[t]
	\caption{Deflate candidate \rSEC{}s}\label{alg:doSECs}
	\begin{algorithmic}[1]
		\Require 
			\Statex SG $\G$, functions $\lb$ and $\ub$ such that for all states $s:$ $\lb(s) \subseteq \val(s) \subseteq \ub(s)$
		\Ensure 
			\Statex	Updated upper bound $\ub'$
		\Procedure{$\doSECs$}{$\game,\lb,\ub$}
		\Statex \Comment{In each MEC, we compute relevant regions, find all candidate $\SEC$s and decrease their upper bounds}
		\State $\ub'(s) \gets \{\vec{0}\}$ for all $\state \in S$ \Comment{Result variable}
		\State $\mathcal M \gets \MECs(\G)$ \Comment{MEC decomposition of the game}
		\For{each $T \in \mathcal M$} \label{line:doSECs:get-mecs}
		\State $\regions \gets  \GETREGIONS(T,\lb)$
		\For {each $R \in \regions$}
		\State $\mathcal{S} \gets \FINDSECS(T,\lb,R)$ \Comment{Candidate SECs} \label{line:find}
		\For {each $\state \in T$}
		\Statex \Comment{If in candidate SEC, deflate~~~~~~~~~~~~~~~~~~~~~~}
			\If {$\state \in C$ for some $C \in \mathcal{S}$} 
				\State \mbox{\hspace*{-2mm}$\ub'(\state) \gets \ub'(\state) \cup (\ub(\state) \cap \exit[\ub](C) \cap R)$} \label{line:deflate}
				\Statex \vspace{-0.1em} \Comment{Otherwise, keep estimate in this region}
			\Else 
				\State \mbox{\hspace*{-2mm}$\ub'(\state) \gets \ub'(\state) \cup (\ub(\state) \cap R)$} \label{line:keepEstimate}
			\EndIf
		\EndFor
		\EndFor
		\EndFor
		\State \Return $\ub'$
		\EndProcedure
	\end{algorithmic}
\end{algorithm}

\begin{algorithm}[t]
	
	\caption{Compute consistent partition into regions}\label{alg:getregions}
	\begin{algorithmic}[1]\Require 
		\Statex MEC $T\subseteq \states$, function $\lb$ such that for all states $s:$ $\lb(s) \subseteq \val(s)$
		\Ensure 
		\Statex	Consistent partition $\regions$ w.r.t. $T$ and $\lb$
		\Procedure{$\GETREGIONS$}{$T,\lb$}
		\State $\regions \gets \{\directions\}$ \Comment{initialize with trivial partition}
		\For{$\state\in T_\circ$} 
			\For {each $B \subseteq \Av(\state)$}
				\State \hspace{-0.6em}$R_B \gets \{\dir \in \directions \mid B = \argmin_{a \in \Av(\state)} \lb(\state,a)[\dir]\}$ \label{line:regionForSet}\vspace{-0.45em}
			\EndFor
			\State  $\regions' = \{R_B \mid B \subseteq \Av(s), R_B \neq \emptyset\}$ 
			\State $\regions \gets $ common refinement of $\regions$ and $\regions'$\label{line:regionCR}
		\EndFor
		\State \Return $\regions$
		\EndProcedure
	\end{algorithmic}
\end{algorithm}

\begin{algorithm}[t]
	\caption{Find candidate \rSEC{}s}\label{alg:findsecs}
	\begin{algorithmic}[1]
		\Require{Under-approximation $\lb$, EC $T$, region $R$ from a consistent partition w.r.t. $T$ and $L$}
		\Ensure{Set of Regional $\SEC$s for $R$, according to $\lb$}
		\Procedure{$\FINDSECS$}{$\lb$,$T$,$R$}
		{$T\subseteq \states,\lb,$ region $R$}
		\State $\dir \gets $ arbitrary element of $R$
		\State $\Av' \gets \Av$
		\For {each $s \in T_\circ$}
		\Statex \Comment{Keep only optimal Minimizer actions~~~~~~~~~~~~~~~~~}
		\State $\Av'(s) \gets \{a \in \Av(s) \mid \lb(s,a)[\dir] =$ 
		\Statex ~~~~~~~~~~~~~~~~~~~~~~~~~~~~~~~~~~~~~~~~~~~~~~~~~~~~~~~$\min_{b \in \Av(s)} \lb(s,b)[\dir]\}$ \label{line:find_secs_actions}
		\EndFor
		\vspace{-0.5em}
		\State \textbf{return} $\MECs(T|_{\Av[\prime]})$ \Comment{\parbox[t]{.5\linewidth}{MEC decomposition on $T$ with actions restricted to~$\Av[\prime]$}}
		\EndProcedure
	\end{algorithmic}
\end{algorithm}

\paragraph{Algorithm~\ref{alg:MOBVI} (\textsf{MO-BVI})} initializes the under- and over-app\-roximations $\lb$ and $\ub$ and updates them using the new Bellman update operator $\mop$. 
This operator first performs the standard Bellman updates and then calls the procedure $\doSECs$, which we exemplified in Section \ref{sec:example}.
The intuition of the whole algorithm is, that as the under-approximation converges, eventually the correct regional SECs are found and deflated. 
When all regional SECs are deflated, the over-approximation approaches the true set of achievable vectors in the limit.
Note that the stopping criterion can be evaluated, as the under- and over-approximation are at all times described by finitely many points, for details see Section~\ref{sec:getReg}.

\paragraph{Algorithm~\ref{alg:doSECs} ($\doSECs$)} is the heart of our new algorithm. It implements the correct handling of end components, ensuring convergence of the upper and the lower approximation to the \emph{same} fixpoint.
As every SEC is an EC and every EC is a subset of a MEC, the algorithm first computes the MEC-decomposition.
Then, for each MEC we compute a consistent partition of the set of directions into regions using Algorithm \ref{alg:getregions}.
Finally, Algorithm \ref{alg:doSECs} updates the over-approximation of every state in the considered MECs.
It does so piece by piece, region by region; this is why in Lines \ref{line:deflate} and \ref{line:keepEstimate} we always intersect with $R$, restricting the update to points in the current region, and take the union with the intermediate result $\ub'$, adding all the points from the previous iterations of the loop over $\regions$.
If a state is part of a regional SEC $C$ (as detected by Algorithm \ref{alg:findsecs}), the upper bound in the current region is reduced to $\exit[\ub](C)$, i.e. to the best exit from the regional SEC.
If a state is not in a candidate SEC for the current region, its upper bound does not change.
Note that the best exit depends on $\ub$, our current best over-approximation.
The intersection with $\ub(s)$ ensures that deflate is monotonic.
Formally, we have the following:
\begin{lemma}[$\DEFLATE$ is monotonic and sound]\label{lem:deflateMonSound}
	Given a game $\G$ with correct upper and lower bounds $\ub$ and $\lb$  (i.e. $\forall s \in S: \lb(\state) \subseteq \val(\state) \subseteq \ub(\state)$), $\ub' = \doSECs(\G,\lb,\ub)$ has the following properties: For all states $s \in S$,
	\begin{itemize}
		\item $\ub'(\state) \subseteq \ub(\state)$ (Monotonicity),
		\item $\val(\state) \subseteq \ub'(\state)$ (Soundness),
	\end{itemize}
\end{lemma}
\begin{proof}[Proof]
	For monotonicity notice that due to line \ref{line:deflate}, $\ub'(s)$ is obtained by intersecting $\ub(s)$ with $\exit[\ub](C)$ on each region $R \in \regions$, which makes sure that $\ub'(\state) \subseteq \ub(\state)$ in the end. For the second item, we have that $\val(\state') \subseteq \ub'(\state')$ for all states $s'$ by assumption. Recall that $\exit[\ub](\val)$ is the set of points achievable from $s \in C$ assuming that Maximizer has control over all states in $C$. Clearly, $\val(s) \subseteq \exit[\val](C) \subseteq \exit[\ub](\val)$, which proves soundness (see \ifarxivelse{Appendix \ref{sec:deflateMonSound}}{\cite[Appendix A.1]{techreport}} for details).
\end{proof}

\paragraph{Algorithm \ref{alg:getregions} ($\GETREGIONS$)} has to return a consistent partition of the set of directions $\directions$, i.e. for all directions in a region, the optimal strategy of Minimizer needs to be the same.
To do that, for every state in the given MEC, we partition the set of directions into regions according to the optimal strategy of Minimizer, i.e. which actions are optimal in the region\footnote{The implementation suggested in Section \ref{sec:getReg} actually computes regions for all orderings of actions. It then describes the regions with the same optimal actions as a union of all regions where these actions are at the top of the ordering.}.
Then we take the \emph{common refinement} of all these partitions.
The common refinement of two partitions $\regions_1$ and $\regions_2$ is defined as the coarsest partition $\regions$ such that for all $R_1 \in \regions_1$, $R_2 \in \regions_2$ we have $R_1 \cap R_2\in\regions$. Notice that the common refinement of any number of consistent partitions (w.r.t. the same $T$ and $\lb$) is again consistent.
Intuitively, in every resulting region the strategy of all Minimizer states in the MEC is constant.
Formally, we have the following lemma:

\begin{lemma}[$\GETREGIONS$ is sound]\label{lem:getRegionsConsistent}
	For any set of states $T$ and bound function $\lb$, the set of regions $\mathcal{R}$ returned by procedure $\GETREGIONS$($T,\lb$) is a consistent partition.
\end{lemma}
\begin{proof}[Proof]
	We simply consider for every subset $B \subseteq \Av(s)$, $s \in T_\circ$, the region $R$ where the actions in $B$ are all optimal. This yields a partition $\regions' = \{R_B \mid B \subseteq \Av(s), R_B \neq \emptyset\}$ which is consistent w.r.t. $\{s\}$ and $\lb$. We repeat this for all $s \in T_\circ$ and take the common refinement of all partitions obtained in this way, yielding a consistent partition for the whole EC $T$ and $\lb$. See the next section on how to technically implement these operations effectively.
\end{proof}

\paragraph{Algorithm \ref{alg:findsecs} ($\FINDSECS$)} is very similar to the single-di\-mensional case (\cite[Alg. 2]{cav18}).
The difference is that in the multi-objective setting we cannot just fix the strategy of Minimizer and compute the ECs in the resulting SG.
We have to pick a direction from the region and consider the strategy of Minimizer w.r.t. that direction.
Since we know that the given region is from a consistent partition by assumption on the input (which is true due to Lemma \ref{lem:getRegionsConsistent}), Minimizer's optimal strategy is the same for all directions in the input region.
Thus the direction can be arbitrarily chosen from that region.
We stress that $\FINDSECS$ is called with the current under-approximation and returns only 
those state sets, which according to the current lower bound form regional SECs; these need not actually be regional SECs according to $\val$.
However, as sketched in the proof of Lemma \ref{lem:deflateMonSound}, deflation is so conservative that it is sound given any EC.
The required property of $\FINDSECS$ is that it eventually finds the correct regional SECs when $\lb$ converges to $\val$ close enough, or formally:

\begin{lemma}[$\FINDSECS$ is sound]\label{lem:findSECSsoundAndCorr}
	For $T \subseteq \states$ and a region $R$ from a consistent partition, it holds that $X \in \FINDSECS(T,\val,R)$ if and only if $X$ is an inclusion-maximal $\rSEC$ for region $R$.
\end{lemma}
\begin{proof}[Proof]
	Since $R$ is from a consistent partition, we can pick any direction $\dir \in R$ and identify Minimizer's optimal actions for the whole region $R$ as in line \ref{line:find_secs_actions}. Let $X$ be a MEC returned by $\FINDSECS$. Then within this EC, Minimizer only has optimal actions for region $R$ and thus, it does not matter how exactly these choices are resolved -- in particular, it does not make a difference if Maximizer takes over control of Minimizer's states as explained earlier. But then, from each $s \in X$, Maximizer can achieve precisely $\exit[\val](X)$. Thus $X$ is an inclusion-maximal \rSEC{} for region $R$.
\end{proof}

\subsection{Effectiveness of $\GETREGIONS$}\label{sec:getReg}

In this section we describe $\GETREGIONS$ in more detail and argue why the computation is effective.
As discussed in Section \ref{sec:example}, regions in our context bear some resemblance to regions of timed automata \cite{DBLP:journals/tcs/AlurD94}. 
We first recall some geometric notions from e.g. \cite{hatcher2002algebraic} that are necessary to talk about the representation of the considered objects:

A \emph{$(k)$-simplex} is a $k$-dimensional polytope given as the convex hull of $k + 1$ affinely independent vertices.
Intuitively, a simplex is a point, line segment, triangle, tetrahedron etc. 
For example, considering Figure \ref{fig:proj-plane}, the point $(0.5,0,0.5)$ is a 0-simplex, the line between this point and $(1,0,0)$ is a 1-simplex, and the whole triangle is a 2-simplex. 
A \emph{face} of a $k$-dimensional simplex is the convex hull of a non-empty subset of the $k+1$ points making up the simplex. 
A \emph{facet} of a $k$-dimensional simplex is a natural face, i.e. a face that uses exactly $k$ points.
For example, for a 2-simplex which is a triangle, the triangle itself, the 3 edges and 3 vertices are all faces.
Each face is also a simplex.
Only the three edges are facets.

A \emph{simplicial complex (SC)} is a set of simplices closed under taking faces, i.e. every face of a simplex in the SC is also part of the SC. 
It also satisfies the property that a non-empty intersection of any two simplices in the SC is a face of both the simplices.
Using Figure \ref{fig:proj-plane} again: Consider the SC containing the two lines (1-simplices) between $(0.5,0,0.5)$ and $(1,0,0)$ as well as between $(0.5,0,0.5)$ and $(0,0,1)$.
It also has to contain the point (0-simplex) $(0.5,0,0.5)$, as that is the intersection of the lines.
Additionally, the points $(1,0,0)$ and $(0,0,1)$ need to be in the SC, as they are the faces of the lines.
In order to represent (i) partitions (disjoint decompositions) and (ii) open regions, as discussed already in Section~\ref{sec:exMOSG}, we consider the open version: we subtract from each simplex all its facets and, abusing the notation, call them simplices and their union SC. 

The invariant of our computation is that all partitions into regions as well as the Pareto frontiers are represented as finite unions of SCs.
The partitions decompose (triangulate) $\directions$, the part of the projective hyperplane that is in the non-negative orthant ($n$-dimensional analog of the first quadrant), which can thus itself be seen as $(n-1)$-simplex; the Pareto frontiers are given by linear functions on the areas defined by regions, hence consists of SCs in the non-negative orthant (of course, generally not arranged in a hyperplane). 
Altogether, since simplices can be stored as the set of their vertices, we can effectively represent these partitions and frontiers by finite sets of finite sets of points.

For the computation of $\GETREGIONS$ on Line \ref{line:regionForSet}, we can compute the intersections of all pairs of Pareto frontiers of the available actions as they are piece-wise linear with finitely many pieces, and we obtain a finite partition.
The projected intersections then become $k$-simplices for some $k > 0$ (the intersection of Pareto frontiers can be points, lines, planes and so on as seen in Fig.~\ref{fig:proj} and \ref{fig:type2}).
Similarly, on Line \ref{line:regionCR}, starting from two finite partitions, their common refinement after the respective triangulation, as e.g. in Figure \ref{fig:triangulation-example}, is also finite and an SC.
Recall the base case for the partition is the SC of the projective $n-1$-simplex.

Finally, the resulting approximation of $\achievable$ is effective since, given a direction $\dir$, we can identify its region and the respective simplex on the Pareto frontiers $\lb$ and $\ub$ and their value in the intersection with $\dir$.
For effectiveness of the stopping criterion on Line \ref{line:StoppingCrit} of Algorithm \ref{alg:MOBVI}, we additionally note that we only need to test for each simplex the differences in its generating points (more precisely the limits as the simplex is open) since the difference is a linear function on each of the finitely many pieces of the approximation.

\section{Correctness Proof of Algorithm MO-BVI}\label{sec:proofStruct}

Our new Bellman operator $\mop$ defined as one application of the loop body of Algorithm \ref{alg:MOBVI} is a higher order operator transforming pairs of the estimate functions: the two estimate functions $\lb,\ub \in \states \to 2^{[0,1]^n}$ for the under-/over-approximation are transformed into a pair with the modified under- and over-approximation.
It can thus be seen as a function of type
\begin{multline*}
\mop:\left(\states  \to 2^{[0,1]^n}\times2^{[0,1]^n}\right) \to
 \left(\states \to 2^{[0,1]^n}\times 2^{[0,1]^n}\right).
 \end{multline*}
 
We fix an SG $\exGame$ and a generalized-reachability objective $\targetsets$ for the following proofs and implicitly use them as parameters of $\mop$.
Note that for all states $\state \in \states$, $\ub<0>(\state) = \dwc(\{\vec{1}\})$ respectively $\lb<0>(\state)=\{\vec{0}\}$ are set by the initialization.

We consider the sequence
$(\lb<i>,\ub<i>) := \mop^i(\lb<0>,\ub<0>)$, $i \in \mathbb{N}$, output by our algorithm.
We also use the notation $\lb<\infty> := \lim_{i \to \infty} \lb<i> := \bigcup_{i \geq 0}\lb<i>$ and $\ub<\infty> := \lim_{i \to \infty} \ub<i> := \bigcap_{i\geq 0} \ub<i>$.

\begin{proposition}{\textbf{Soundness}}
\label{prop:soundness}
\newline Algorithm \ref{alg:MOBVI} computes for each state $\state \in \states$ a sequence of monotonic over- and under-approximations of $\achievable(\state)$, i.e.
$\forall i \in \mathbb{N}: \lb<i>(\state) \subseteq \achievable(\state) \subseteq \ub<i>(\state)$ and for $i < j, \lb<i>(\state) \subseteq \lb<j>(\state)$ as well as $\ub<i>(\state) \supseteq \ub<j>(\state)$.
\end{proposition}

\begin{proposition}{\textbf{Convergence from below}}
\label{prop:conv_below}
\newline $\forall$ states $\state \in \states$ and all directions $\dir \in \directions: \lb<\infty>(\state)[\dir] = \achievable(\state)[\dir]$.
\end{proposition}

\begin{proposition}{\textbf{Convergence from above}}
\label{prop:conv_above}
\newline $\forall$ states $\state \in \states$ and all directions $\dir \in \directions: \ub<\infty>(\state)[\dir] = \achievable(\state)[\dir]$.
\end{proposition}

Note that for all directions $\dir$ and for all $\state \in \states$ by definition $\achievable(\state)[\dir] = \pareto(\state)[\dir]$.
Using this and the three propositions, we can prove the main theorem.
\begin{theorem}
Algorithm \ref{alg:MOBVI} computes convergent monotonic over- and under-approximations of $\pareto(\state)$ for each $\state \in \states$.
Since it is convergent, for every $\epsilon >0$ there exists an $i$, such that for every $\state \in \states$ and direction $\dir \in \directions: \ub<i>(\state)[\dir] - \lb<i>(\state)[\dir] < \epsilon$. 
So by instantiating $\state$ with $\initstate$, we solve the problem posed in Section \ref{sec:problem}.
\end{theorem}

\begin{proof}[Proof of Propositions \ref{prop:soundness} and \ref{prop:conv_below}]
Note that for all $i \in \mathbb{N}$ it holds that $\lb<i> = \bellman^i(\lb<0>)$, since $\doSECs$ does not change the under-approximation.
\cite[Proposition 8]{DBLP:journals/iandc/BassetKW18} proves that $\bellman$ is order-preserving, i.e. monotonic, and that 
it converges to the unique least fixpoint $\achievable$ when repeatedly applied to the bottom element of a complete partial order. 
The least possible lower bound assigns $\vec{0}$ to all $\states$, since there is no smaller vector that can be assigned to a state.
This is exactly the definition of $\lb<0>$, which implies that for all $\state \in \states$, the closure of $\lb<\infty>(\state)$ equals the closure of $\achievable(\state)$, which implies Proposition \ref{prop:conv_below}.

For the soundness of the over-approximation we require that the additional operation, namely $\doSECs$, performed by $\mop$ is sound (proven in Lemma \ref{lem:deflateMonSound}).
The monotonicity of the under- and over-approximation follows from the monotonicity of $\bellman$~\cite[Proposition 8]{DBLP:journals/iandc/BassetKW18} and of $\doSECs$ (Lemma~\ref{lem:deflateMonSound}) 
Thus we can deduce Proposition \ref{prop:soundness}.
\end{proof}

It only remains to show Proposition \ref{prop:conv_above}. As a key ingredient for the proof we will use the following:
\begin{lemma}[Fixpoint]
	\label{lem:fixpoint}
	$\mop(\lb<\infty>,\ub<\infty>) = (\lb<\infty>,\ub<\infty>)$, i.e. the limit of $\mop$ is also a fixpoint.
\end{lemma}
\begin{proof}[Proof idea]
	We only need to argue about the second component $\ub<\infty>$. If we did not have a fixpoint, then a further application of $\mop$ would find a SEC $T$ for some region $R$ and decrease the over-approximation $\ub<\infty>$ for some $\dir \in R$ and $s \in T$, i.e. $\exit[\ub<\infty>](T)[\dir] < \ub<\infty>(s)[\dir]$. The key idea is that since the lower approximations $L_i$ converge to $\val$, the \rSEC{} $T$ is detected and deflated infinitely many times before convergence. But this means that $\exit[\ub<\infty>](T)[\dir] = \ub<\infty>(s)[\dir]$, contradiction.
	For more details, see 
	\ifarxivelse{Appendix \ref{app:fixpoint}}{\cite[Appendix A.2]{techreport}}.
\end{proof}

\begin{proof}[Proof of Proposition \ref{prop:conv_above}]
We will use the fixpoint property from Lemma \ref{lem:fixpoint} to derive a contradiction.
We assume for contradiction that there is a state $\state \in \states$ and a direction $\dir \in \directions$ such that $\ub<\infty>(\state)[\dir] \neq \achievable(\state)[\dir]$.
Applying the Bellman operator once more to $(\lb<\infty>,\ub<\infty>)$ results in a new upper bound $\ub'$. We will show that $\ub' \subsetneq \ub<\infty>$.
In other words, applying the loop once more decreases the over-approximation.
This is a contradiction to $\ub<\infty>$ being a fixpoint and proves our goal.
\begin{enumerate}
	\item Assume for contradiction, that $\exists t \in \states, \dir \in \directions: \ub<\infty>(t)[\dir] \neq \achievable(t)[\dir]$ and thus $\exists \dir\ \ub<\infty>(t)[\dir] > \achievable(t)[\dir]$ with Prop. \ref{prop:soundness}. We fix this direction $\dir$ and $t$ for the rest of the proof.\label{step:fix}
	
	\item Let $X := \set{\state \in \states \mid \Delta(\state) = \max_{t \in \states} \Delta(t)}$, 
		where $\Delta(\state) := \ub<\infty>(\state)[\dir] - \achievable(\state)[\dir]$ is the difference between 
		over-approximation $\ub<\infty>(\state)$ and achievable set $\val$ in $\dir$. 
		\begin{enumerate}
			\item We also define $\Delta(\state,\action) := \ub<\infty>(\state,\action)[\dir] - \achievable(\state,\action)[\dir]$ for an action $\action \in \Av(s)$.
			
			\item By assumption, $X \neq \emptyset$ and for all $\state \in X: \Delta(\state) > 0$.
		\end{enumerate}
		Note: $\Delta(s), \Delta(s,a)$ and $X$ are all defined w.r.t. the fixed direction $\dir$ (not indicated in notation to avoid clutter).
	\item $\Delta(s) > 0$ implies that $\dwc(\mathbbm 1_{\targetsets}(\state))[\dir] = 0$, i.e. $s$ is not contained in target sets ``aligned'' in direction $\dir$ because otherwise, $ \ub<\infty>(\state)[\dir] = \achievable(\state)[\dir] = \dwc(\mathbbm 1_{\targetsets}(\state))[\dir]$.\label{step:dwc1_obs}
	\item \label{step:contract} For all $(\state,\action) \leaves X$ it holds that $\Delta(\state,\action) < \Delta(\state)$.\\
	\textbf{\textit{Reason:}}
	If $(\state,\action) \leaves X$, then $\exists \state' \in \post(\state,\action) \setminus X$. Note that $\Delta(\state') < \Delta(\state)$ by construction of $X$
	\begin{align*}
	\Delta(\state,\action) &= \ub<\infty>(\state,\action)[\dir] - \val(\state,\action)[\dir] \tag{Definition of $\Delta(s,a)$}\\
	&= \sum_{\state' \in \post(\state,\action)}  \delta(s,a,s') \left(\ub<\infty>(\state')[\dir] - \val(\state')[\dir]\right) \tag{Definition of $\val(\state,\action)$ and Step \ref{step:dwc1_obs}}\\
	&= \sum_{\state' \in \post(\state,\action)}  \delta(s,a,s')\Delta(s') \tag{Definition of $\Delta(s)$}\\
	&<\Delta(\state) \tag{since  $t \in \post(\state,\action)$ and $\Delta(t) < \Delta(\state)$}
	\end{align*}
	
	\item No state in $X$ depends on a leaving action. Formally:
		\begin{enumerate}[(a)]
		\item 
		$\forall \state \in X_\circ, (\state,\action) \leaves X: \val(\state)[\dir] < \val(\state,\action)[\dir]$,
		i.e. a leaving action leaving $X$ cannot be optimal for Minimizer in direction $\dir$.\\
		\textbf{\textit{Reason:}}
		Since $\state$ is a state of Minimizer, $\forall \action \in \Av(s): \ub<\infty>(\state)[\dir] \leq \ub<\infty>(\state,\action)[\dir]$.
		From this and the inequality from the previous Step \ref{step:contract}, we get that:
		\begin{align*}
			&\ub<\infty>(\state,\action)[\dir] - \val(\state,\action)[\dir]\\
			&=\Delta(\state,\action) \tag{Definition of $\Delta$}\\
			&< \Delta(\state)  \tag{Step \ref{step:contract}}\\
			&= \ub<\infty>(\state)[\dir] - \val(\state)[\dir] \tag{Definition of $\Delta$}\\
			&\leq \ub<\infty>(\state,\action)[\dir] - \val(\state)[\dir] \tag{$\state \in X_\circ$}
		\end{align*}
		Subtracting $\ub<\infty>(\state,\action)[\dir]$ and multiplying by (-1) yields the claim.
		\item $\forall \state \in X_\Box$, $\ub<\infty>(s)[\dir] > \sum_{a \in \Av(s,a)}w_a \cdot \ub<\infty>(s,a)$ if $w_a > 0$ for some action $a$ exiting $X$. Intuitively, this means that Maximizer cannot assign positive weight to any action leaving $X$. The proof is similar to part (a).
		\end{enumerate} \label{step:not_depend_outside}
		
	\item $X$ contains an EC because if not, then $\exists s \in X: \forall \action \in \Av(s): (s,a) \leaves X$. But then $s$ necessarily depends on a leaving action in the sense of the previous Step \ref{step:not_depend_outside}, contradiction.
	
	\item Using that $X$ contains an EC, we can show that $X$ even contains a regional \emph{simple} EC  $Z \subseteq X$ w.r.t. to the region $\set{\dir}$. \label{step:fact2}
	Applying $\mop$ once more to $(\val,\ub<\infty>)$, the over-approximation decreases.\\
	
		\textbf{\textit{Reason:}} We only give high-level intuition here, as the proof is very technical. The formal details are in 
		\ifarxivelse{Appendix \ref{app:fact2}}{\cite[Appendix A.3]{techreport}}.
		We prove by a large case distinction that $X$ contains a regional SEC $Z$ for the region $\set{\dir}$.
		Since $\lb<\infty>=\val$, by Lemma \ref{lem:findSECSsoundAndCorr} this regional SEC is found and deflated.
		By construction of $Z$, then its value is set to ``depend on the outside'', i.e. it assigns a positive weight on an action leaving $X$.
		Then, by Step \ref{step:not_depend_outside}, the over-approximation is reduced and we arrive at a contradiction.\qedhere
\end{enumerate}
\end{proof}

\section{Conclusion} \label{sec:conc}

For a given $\varepsilon>0$ and a generalized-reachability stochastic game, we compute an $\varepsilon$-approximation of its Pareto frontier.
Our algorithm can be run as an anytime algorithm, reporting the under- and over-approximations on the frontier, due to an extended version of value iteration.
We have suggested the name ``bounded value iteration'' as it better generalizes to higher dimensions than ``interval iteration''.
We conjecture that this technique can be generalized to other models, such as concurrent games, and more complex objectives, such as total reward.
Finally, while decidability remains open, the approximation algorithms are practically more relevant even in the single-dimensional case.
Note that approximative value iteration is the default technique for analysis of MDP, although there is an exact and polynomial solution by linear programming.
The reason is that the theoretical worst-case complexity of value iteration is practically not too relevant.
Consequently, an efficient implementation, possibly exploring only a part of the state space using learning, as e.g.\ in \cite{atva,cav18}, may be an interesting future direction.

\begin{acks}
	Pranav Ashok, Jan K\v{r}et{\'i}nsk{\'y} and Maximilian Weininger were funded in part by TUM IGSSE Grant 10.06 (PARSEC) and the German Research Foundation (DFG) project KR 4890/2-1 ``Statistical Unbounded Verification''. Krishnendu Chatterjee was supported by the ERC CoG 863818 (ForM-SMArt) and Vienna Science and Technology Fund (WWTF) Project ICT15-003. Tobias Winkler was supported by the RTG 2236 UnRAVeL.
\end{acks}


\bibliography{ref}
\ifarxivelse{
\appendix
\section*{Appendix}
\section{Technical proofs}

We fix an SG $\exGame$ and a generalized-reachability objective $\targetsets$ for the following proofs and implicitly use them as parameters of $\mop$.

\subsection{Proof of Lemma \ref{lem:deflateMonSound}}
\label{sec:deflateMonSound}

\paragraph{Soundness}
To prove soundness of $\DEFLATE$, we first need two auxiliary lemmata.

\begin{lemma}[$\exit$ for a Maximizer state is correct] \label{lem:casstate}
	If a state $\state \in \states<\Box>$ belongs to the Maximizer, then $\exits[\val](\set{\state}) = \val(\state)$.
\end{lemma}
\begin{proof}
	If we show that for all Maximizer states $\state$, $\exit[\val](\set{\state}) = \bellman(\val)(\state)$, then the lemma holds because $\val = \bellman(\val)$ as $\val$ is a fixpoint of $\bellman$.
	\begin{itemize}
		\item If $\state$ is a non-target Maximizer state, then $\exits[\val](\set{\state})$ =
		$\conv(\bigcup_{(\state, \action) \in \exits(\set{\state})} \val(\state,\action))$ = $\conv(\bigcup_{\action \in \Av(\state)} \val(\state,\action))$ = $\bellman(\val)(\state)$. 
		While $\Av(\state)$ may contain a self-loop action which is not contained in $\exits(\set{\state})$, this does not matter as the Maximizer cannot improve its value by choosing a self-loop action unless $\state$ is a target. Hence, adding a $\val(\state, \action')$ term, where $\action'$ is a self-loop, to the inner union operation does not change the result.
		\item If the Maximizer state $\state$ is a target, then the Bellman operator sets the corresponding direction to 1. This is what the first term ($\sum_{\state' \in \{\state\}} \{\mathbbm 1_{\targetsets}(\state')\} = \{\mathbbm 1_{\targetsets}(\state)\}$) in the definition of best exit does:
		Afterwards the reasoning follows the same arguments as in the previous case.
	\end{itemize}
\end{proof}

\begin{lemma}[$\exit$ for a set of states is an over-approximation] \label{lem:casset} 
	Given an $\EC\ T$, and a correct upper bound $\ub$ with $\ub(\state) \supseteq \val(\state)$ for all $\state \in \states$, we get that
	$\forall \state \in T: \exit[\ub](T) \supseteq \val(\state)$
\end{lemma}
\begin{proof}
	Let us introduce a new Maximizer state $t$ representing $T$. 
	Let $t$ be in all target sets that some state of $T$ is in, i.e. $(\exists \state \in T. \state \in \targetsets_i) \implies t \in \targetsets_i$ for all target sets $\targetsets_i \in \targetsets$.
	Let $\Av(t) = \exits(T)$.
	Since $t$ can randomize between any set of actions that any of the states in $T_\Box$ can choose, $\forall s \in T_\Box: \val(t) \supseteq \val(s)$. 
	Moreover, $\forall s \in T_\circ,\ \exists s' \in T_\Box:\ \val(s) \subseteq \val(s')$.
	If this was not the case, it means that there exists some Minimizer state $s_\circ$ that has a value greater than all Maximizer states. 
	Since $T$ is an $\EC$, $s_\circ$ has an action $a_\circ$ whose successors are all in the $\EC$. 
	This implies that $\val(s_\circ, a_\circ)$ cannot be greater that $\val(s')$ for all $s' \in T$.
	
	Using Lemma \ref{lem:casstate} and the fact that $\ub$ is a correct upper bound, we get that $\exit[\ub](T) = \exit[\ub](\set{t}) \supseteq \exit[\val](\set{t}) = \val(t)$.
	Combining this with the previous argument yields $\forall s~\in~T:\ \exit[\ub](T) \supseteq \val(s)$.
	
\end{proof}

\begin{lemma}[$\doSECs$ is sound]\label{lem:do-secs-sound}
	For correct upper and lower bound functions $\ub$ and $\lb$ with $\lb(\state) \subseteq \val(\state) \subseteq \ub(\state)$, for each $\state \in \states$, it holds that $\ub' = \doSECs(\G, \lb,\ub)$  is still correct, i.e. $\val(\state) \subseteq \ub'(\state)$ for all $\state \in \states$.
\end{lemma}
\begin{proof}
	The new over-approximation $\ub'$ is constructed region by region, and for each region the update is either performed by applying Line \ref{line:deflate} or Line \ref{line:keepEstimate} of $\doSECs$.
	We will argue for both the lines that after applying the update, $\ub'$ is still an over-approximation in this region.
	Then, using the additional fact that by Lemma \ref{lem:getRegionsConsistent}, the disjoint union of the set of regions $\doSECs$ considers covers the whole space of directions, we get that $\ub'$ is an over-approximation in every direction.
	
	If we apply Line \ref{line:keepEstimate}, the previous $\ub'$ only contains points in other regions; the intersection with $R$ ensures that this update affects only the current region.
	Thus, the only remaining term is $\ub(s)$, and clearly $\ub(s) \supseteq \ub(s)$.
	
	If we apply Line \ref{line:deflate}, the argument is the same except for the term $\exit[\ub](C)$. However, we know that $\state \in C$, and hence by Lemma \ref{lem:casset} we have that $\ub'(s) \supseteq \exit[\ub](C) \supseteq \val(\state)$.
\end{proof}


	

\subsection{Proof of Lemma \ref{lem:fixpoint}}\label{app:fixpoint}

\begin{proof}
	By the previous considerations we have $\mop(\lb<\infty>,\ub<\infty>) = \mop(\val,\ub<\infty>) = (\bellman(\val),\ub') = (\val,\ub')$ for some $\ub'$. Thus it only remains to show that $\ub<\infty> = \ub'$. If $\ub<\infty> \neq \ub'$, then when executing $\mop(\val,\ub<\infty>)$, a \rSEC{} $T$ for some region $R$ is found and deflated on this region. The over-approximation is thereby decreased, i.e. $\ub'(s)[\dir] = \exit[\ub<\infty>](T)(s)[\dir] < \ub<\infty>(s)[\dir]$ for some $s \in T$ and $\dir \in R$. However, already during the iterations $\mop(\lb<0>,\ub<0>), \mop(\lb<1>,\ub<1>),\ldots$ the \rSEC{} $T$ is detected infinitely often because eventually, the under-approximations $\lb<i>$ are sufficiently close to the true sets $\val$ of achievable points.
	Thus for infinitely many $j\geq 0$, the algorithm sets $\ub<j>(s)[\dir] = \exit[\ub<j>](T)(s)[\dir]$. But this implies that $\ub<\infty>(s)[\dir] = \exit[\ub<\infty>](T)(s)[\dir]$, contradiction.
	In special cases, it may happen that only a sub-EC $T' \subseteq T$ is detected infinitely often, e.g. if there are multiple optimal Minimizer-actions whose lower approximations converge at different rates. However, this is not a problem because the states in $T \setminus T'$ will adjust their values correspondingly via the standard Bellman-update $\bellman$.
\end{proof}

\subsection{Proof of Step \ref{step:fact2} of Proposition \ref{prop:conv_above}}\label{app:fact2}

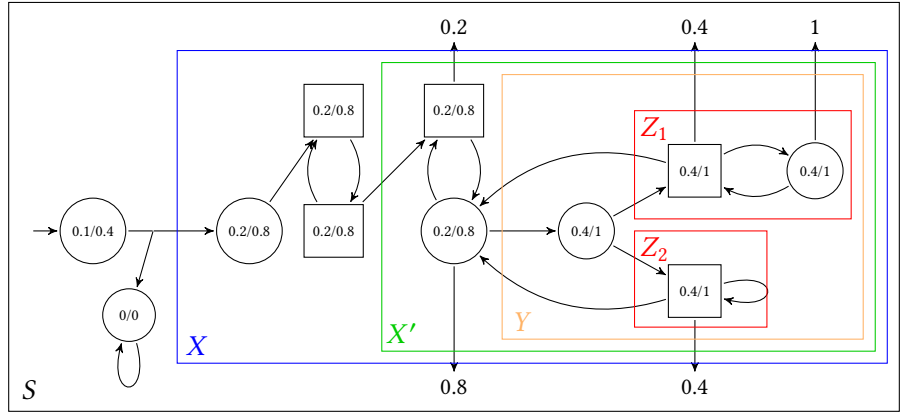
\begin{figure*}[t]

\centering
\begin{tikzpicture}[scale=1.6,font=\tiny]

\drawdummy (init) at (-0.5,0) {};
\drawdummy (prob) at (0.5,0) {};
\drawcirc (0) at (0,0) {0.1/0.4};
\drawcirc (1) at (1.3,0) {0.2/0.8};
\drawbox (2) at (2,1) {0.2/0.8};
\drawbox (3) at (2,0) {0.2/0.8};
\drawbox (4) at (3,1) {0.2/0.8};
\drawcirc (5) at (3,0) {0.2/0.8};
\drawcirc (6) at (4.1,0) {0.4/1};
\drawcirc (7u) at (6,0.5) {0.4/1};
\drawbox (8u) at (5,0.5) {0.4/1};
\drawbox (7d) at (5,-0.5) {0.4/1};

\drawcirc (E1) at (0.3,-0.7) {0/0};
\drawnum (E2) at (3,1.7) {\small 0.2};
\drawnum (E8) at (3,-1.3) {\small 0.8};
\drawnum (E1u) at (6,1.7) {\small 1};
\drawnum (E4u) at (5,1.7) {\small 0.4};
\drawnum (E4d) at (5,-1.3) {\small 0.4};

\draw[->] (init) to (0);
\draw[-] (0) to (prob);
\draw[->] (prob) to (E1);
\draw[->] (prob) to (1);
\draw[->] (1) to (2);
\draw[->] (2) to[bend left] (3);
\draw[->] (3) to[bend left] (2);
\draw[->] (3) to (4);
\draw[->] (4) to (E2);
\draw[->] (5) to (E8);
\draw[->] (4) to[bend left] (5);
\draw[->] (5) to[bend left] (4);
\draw[->] (5) to (6);
\draw[->] (6) to (8u);
\draw[->] (6) to (7d);
\draw[->] (7d) to[bend left] (5);
\draw[->] (7d) to[loop right] (7d);
\draw[->] (7d) to (E4d);
\draw[->] (7u) to[bend left] (8u);
\draw[->] (8u) to[bend left] (7u);
\draw[->] (7u) to (E1u);
\draw[->] (8u) to (E4u);
\draw[->] (8u) to[bend right] (5);
\draw[->] (E1) to[loop below] (E1);

\draw [black] (-0.7,-1.5) rectangle (6.7,1.9);
\node (S) at (-0.53,-1.35) {\large $\states$};
\draw [blue] (0.7,-1.1) rectangle (6.6,1.5);
\node (X) at (0.87,-0.95) {\large \textcolor{blue}{$X$}};
\draw [green!80!black] (2.4,-1) rectangle (6.5,1.4);
\node (X') at (2.57,-0.85) {\large \textcolor{green!80!black}{$X'$}};
\draw [orange!60] (3.4,-0.9) rectangle (6.4,1.3);
\node (Y) at (3.57,-0.75) {\large \textcolor{orange!60}{$Y$}};
\draw [red] (4.5, 0.1) rectangle (6.3,1);
\node (Z1) at (4.67,0.85) {\large \textcolor{red}{$Z_1$}};
\draw [red] (4.5, -0.8) rectangle (5.6,0);
\node (Z2) at (4.67,-0.15) {\large \textcolor{red}{$Z_2$}};

\end{tikzpicture}
\caption{An example of an SG where all the sets in the proof of Step \ref{step:fact2}.
	Numbers on exiting arrows denote $\val$$[\dir]$, similar to the example in Section \ref{sec:example}.
	Every state is inscribed with $\val$$[\dir] / \ub$$[\dir]$, where $\ub$ is an upper bound that has not yet converged, but is a fixpoint of $\bellman$.
}
\label{fig:XYZ}
\end{figure*}

\begin{proof}
	We have the context of the proof of Proposition \ref{prop:conv_above}, in particular we know that $X \subseteq \states$ contains an EC and that for all states $\state \in X: \Delta(\state) = \max_{\state \in \states} \Delta(\state)=:\dif$.
	
	We now need several case distinctions to finally find a regional SEC $Z$, because states can have a large difference just by depending on a SEC.
	See Figure \ref{fig:XYZ} for an example of a SG where all the sets we introduce in the following are different.
	For the sake of clarity we only consider a single-dimension in the SG, as we have fixed a direction.

	\begin{enumerate}
		\item Let $X' \subseteq X$ be a bottom MEC in $X$.
		\paragraph*{Reasoning}
		A bottom MEC is a MEC that has no exits. They are computed by computing the MEC decomposition of $X$, ordering them topologically and picking one at the end of a chain.
		Note that we only require that $X'$ is a bottom MEC in the game restricted to $X$, not a bottom MEC considering all states $\states$.
		$X'$ exists, since there is an EC in $X$, so there also is at least one MEC in $X$.
		\item Let $m = \max_{\state \in X'} \ub(\state)[\dir]$ be the maximal upper bound in $X'$.
		\item Let $Y := \set{\state \mid \state \in X' \wedge \ub<\infty>(\state)[\dir] =m}$ be the states with maximal upper bound in $X'$.
		\item $\forall \state \in Y, \exists \action \in \Av(\state): \neg (\state,\action) \leaves Y$, i.e. all states in $Y$ have actions that stay in $Y$. \label{step:stayY}
		\paragraph*{Reasoning}
		We prove this by a case distinction over where the actions of $Y$ can exit to. Let $\state \in Y$ be an arbitrary state. 
		There has to be some convex combinations of actions that it achieves $m$ in direction $\dir$.
		We say that an action $\action<m>$ exits $X$ towards a set of states $T$, if $(\state,\action<m>) \leaves~X$ and some successor of the action is in $T$.
		\begin{itemize}
			\item We cannot put weight on an action that exits towards $\states \setminus X$.
			Otherwise, $\state$ would ``depend on the outside'' and we get a contradiction by Step \ref{step:not_depend_outside} of the proof of Proposition \ref{prop:conv_above}.
			\item We cannot put weight on an action that exits towards $X \setminus X'$, because $X'$ is a bottom MEC in $X$. 
			If an action left towards some state $t \in X \setminus X'$, then from $t$ there would be no reachable EC in $X$.
			Thus from $t$ we eventually have to exit $X$, as the play cannot remain in a transient part. This is a contradiction, as then some state after $t$ has to ``depend on the outside''.
			\item We cannot put weight on an action that exits towards $X' \setminus Y$, as by definition of $Y$ all states $\state' \in X' \setminus Y$ have $\ub<\infty>(\state')[\dir] < m$, and thus would get a smaller number.
			\item The only remaining possibility is that we only put weight on actions that stay in $Y$. Thus, every state needs to have at least one action that stays in $Y$.
		\end{itemize}
		\item Let $Z$ be a bottom MEC in $Y$.
		\paragraph*{Reasoning}
		This works as when finding $X'$ in $X$. $Z$ exists, since by the previous step all states in $Y$ have actions staying in $Y$, and hence there has to be an EC.
		\item For all states $\state \in Z: \val(\state)[\dir] = m - \dif$.
		\paragraph*{Reasoning}
		Since $Z \subseteq Y$, $\ub<\infty>(\state)[\dir] = m$ and
		since $Y \subseteq X$, $\Delta(\state) = \dif$.
		We get the following chain of equations:
		$\dif = \Delta(\state) = \ub<\infty>(\state)[\dir] - \val(\state)[\dir] = m - \val(\state)[\dir]$.
		Reordering yields the statement.
		\item Thus, $Z$ is an $\SEC$ for region $\set{\dir}$.
		\paragraph*{Reasoning}
		All states in $Z$ have the same value in this direction, and there has to be an exit.
		So some state can take an exit. All states need to be able to take the same convex combination of all exits, because if Minimizer was able to restrict Maximizer from doing so, the states would have different values.
		\item When applying $\mop$ once more, $Z \in \mathcal{S}$ in Line \ref{line:find} of Algorithm \ref{alg:doSECs}. \label{step:isFound}
		\paragraph*{Reasoning}
		Some $X'' \in MEC(\G)$ with $X'' \supseteq X'$, by definition of $X'$.
		So $\mathcal{R} \gets \GETREGIONS(X'',\lb<\infty>)$ is executed. 
		Since $\biguplus_{R \in \mathcal{R}} R = \directions$ by Lemma \ref{lem:getRegionsConsistent}, there is some $R \in \mathcal{R}$ with $\dir \in R$.
		Also by that Lemma we have that the relative order of exits for all directions in $R$ is the same, and since it was called with $\lb<\infty>$, it is correct. 
		Thus, we can apply Lemma \ref{lem:findSECSsoundAndCorr}, which proves the statement.
		\item $\exit[\ub<\infty>](Z)[\dir] < m$ \label{step:smaller}
		\paragraph*{Reasoning}
		$\exit[\ub<\infty>](Z)[\dir]$ must put positive weight on some exit of $Z$. 
		If it puts weight on an action leaving $X$, it ``depends on the outside'' and by Step \ref{step:not_depend_outside} of the proof of Proposition \ref{prop:conv_above}, $\exit[\ub<\infty>](Z)[\dir] < m$.
		The only other possible exit is to $X' \setminus Y$ because $Z$ is a bottom MEC in $Y$.
		For all states $\state' \in X' \setminus Y$, it holds that $\ub(\state')[\dir] < m$.
		If $\exit[\ub<\infty>](Z)[\dir]$ is constructed from a convex combination of exits only to $X' \setminus Y$, then also $\exit[\ub<\infty>](Z)[\dir]< m$.
		\item $\forall \state \in Z: \mop(\ub<\infty>)(\state)[\dir] = \exit[\ub<\infty>](Z)[\dir]$
		\paragraph*{Reasoning}
		Let $\state \in Z$.
		The upper bound is modified by Line \ref{line:deflate}.
		Since $\dir \in R$ by how the algorithm found $Z$ (Step \ref{step:isFound}) and since $\exit[\ub<\infty>](Z)[\dir] < m = \ub(\state)[\dir]$, the new upper bound is exactly $\exit[\ub<\infty>](Z)[\dir]$ for each $\state \in Z$.
		\item Thus, by combining the previous two steps, we finally arrive at a contradiction (to the fixpoint property of $\ub<\infty>$, Lemma \ref{lem:fixpoint}), 
		since $\forall \state \in Z: \mop(\ub<\infty>)(\state)[\dir] = \exit[\ub<\infty>](Z)[\dir] < \ub<\infty>(\state)[\dir]$
		
	\end{enumerate}
\end{proof}
}
{}

\end{document}